\documentclass[11pt,a4wide]{article}

\usepackage[warning, all]{onlyamsmath}

\usepackage{latexsym,fullpage,amsmath,amssymb,amstext,url,verbatim}
\usepackage{epsfig}
\usepackage{graphics,graphicx}
\usepackage{amsthm}
\usepackage{cite} 
\usepackage{float}
\floatname{algorithm}{Procedure}
\usepackage[ruled,longend,vlined]{algorithm2e}

\usepackage[color=green!40]{todonotes}

\newcommand{\name}[1]{\textsc{#1}}

\newcommand{\KD}{\name{$K_4$-minor cover}}
\newcommand{\KC}{{$K_4$-minor cover}}
\newcommand{\KCC}{\name{disjoint $K_{4}$-minor cover}}

\newcommand{\mMSO}{\name{$p$-min-MSO}}
\newcommand{\YES}{\name{YES}}

\newcommand{\parfuncdefn}[4]{
\begin{tabbing}
\name{#1}\\
\emph{Input:} \hspace{1.2cm} \= \parbox[t]{14cm}{#2} \\
\emph{Parameter:}            \> \parbox[t]{14cm}{#3} \\
\emph{Output:}             \> \parbox[t]{14cm}{#4} \\
\end{tabbing}
}

\newcommand{\mc}{\mathcal}

\newtheorem{lemma}{Lemma}
\newtheorem{claim}{Claim}

\newtheorem{definition}{Definition}
\newtheorem{theorem}{Theorem}

\newtheorem{observation}{Observation}
\newtheorem{reduction}{Reduction Rule}
\newtheorem{branching}{Branching Rule}
\newtheorem{numclaim}{Claim}

\newenvironment{proofof}{\noindent \textit{Proof of claim. }}{\hfill$\Diamond$}
\newenvironment{reminder}[1]{\smallskip \noindent {\bf Reminder of #1  }\em}{}

\begin{document}

\title{\textbf{A single-exponential FPT algorithm \\for the \KD{}
    problem}\thanks{This work is supported by the ANR project
    AGAPE (ANR-09-BLAN-0159).}}

\author{Eun Jung Kim\footnotemark[1], Christophe
  Paul\footnotemark[2], Geevarghese Philip\footnotemark[3]}

\date{}

\footnotetext[1]{CNRS, LAMSADE, Paris,
  France.\texttt{\small{\{eunjungkim78\}@gmail.com}}}

\footnotetext[2]{CNRS, LIRMM, Montpellier,
  France.\texttt{\small{\{paul\}@lirmm.fr}}}

\footnotetext[3]{MPII, Saarbr\"{u}cken, Germany.\texttt{\small{gphilip@mpi-inf.mpg.de}}}

\maketitle

\vspace{-0.5cm}
\begin{abstract}
Given an input graph $G$ and an integer $k$, the parameterized \KD{} problem asks whether there is a set $S$ of at most $k$ vertices whose deletion results in a $K_4$-minor-free graph, or equivalently in a graph of treewidth at most $2$. This problem is inspired by two well-studied parameterized vertex deletion problems, \textsc{Vertex Cover} and \textsc{Feedback Vertex Set}, which can also be expressed as \textsc{Treewidth-$t$ Vertex Deletion} problems: $t=0$ for {\sc Vertex Cover} and $t=1$ for {\sc Feedback Vertex Set}.
While a single-exponential FPT algorithm has been known for a long time for \textsc{Vertex Cover}, such an algorithm for \textsc{Feedback Vertex Set} was devised comparatively recently. While it is known to be unlikely that \textsc{Treewidth-$t$ Vertex Deletion} can be solved in time $c^{o(k)}\cdot n^{O(1)}$, it was open whether the \KD{} could be solved in single-exponential FPT time, i.e. in $c^k\cdot n^{O(1)}$ time. This paper answers this question in the affirmative.
\end{abstract}

\vspace{-0.4cm}
\section{Introduction}


\vspace{-0.2cm}
Given a set $\mathcal{F}$ of graphs, the parameterized \textsc{$\mathcal{F}$-minor cover} problem is to identify a set $S$ of at most $k$ vertices --- if it exists --- in an input graph $G$ such that the deletion of $S$ results in a graph which does not have any graph from $\mathcal{F}$ as a minor; the parameter is $k$. Such a set $S$ is called an \textit{$\mathcal{F}$-minor cover} (or an \textit{$\mathcal{F}$-hitting set}) of \(G\).  A number of fundamental graph problems can be viewed as \textsc{$\mathcal{F}$-minor cover} problems. Well-known examples include \textsc{Vertex Cover} ($\mathcal{F}=\{K_2\}$), \textsc{Feedback Vertex Set} ($\mathcal{F}=\{K_3\}$), and more generally the \textsc{Treewidth-$t$ Vertex Deletion} for any constant $t$, which asks whether an input graph can be converted to one with treewidth at most $t$ by deleting at most $k$ vertices. Observe that for $t=0$ and $1$, \textsc{Treewidth-$t$ Vertex Deletion} is equivalent to \textsc{Vertex Cover} and \textsc{Feedback Vertex Set}, respectively.
The importance of  \textsc{Treewidth-$t$ Vertex Deletion} is not only theoretical. For example, even for small values of $t$, efficient algorithms for this problem would improve algorithms for inference in Bayesian Networks as a subroutine of the \emph{cutset conditioning method}~\cite{BidyukD07}. This method is practical only with small value $t$ and efficient algorithms for small treewidth $t$, though not for any fixed $t$, are desirable.

In this paper we consider the parameterized \textsc{$\mathcal{F}$-minor cover} problem for $\mathcal{F}=\{K_4\}$, which is equivalent to the \textsc{Treewidth-$2$ Vertex Deletion}. The NP-hardness of this problem is due to~\cite{LewisYannakakis1980}. Fixed-parameter tractability (\textit{i.e.} can be solved in time $f(k)\cdot n^{O(1)}$ for some computable function $f$) follows from two celebrated meta-results: the Graph Minor Theorem of Robertson and Seymour~\cite{RS04} and Courcelle's theorem~\cite{Cou90}. Unfortunately, the resulting algorithms involve huge exponential functions in $k$ and are impractical even for small values of $k$.

In recent years, single-exponential time parameterized algorithms --- those which run in $c^k\cdot n^{O(1)}$ time for some constant \(c\) --- and also sub-exponential time parameterized algorithms have been developed for a wide variety of problems. Of special interest is the bidimensionality theory introduced by Demaine et al.~\cite{DFH05} as a tool to obtain sub-exponential parameterized algorithms for the so-called bidimensional problems on \(H\)-minor-free graphs.  It is also known to be unlikely that \emph{every} fixed parameter tractable problem can be solved in sub-exponential time~\cite{CCF05}. For problems which probably do not allow sub-exponential time algorithms, ensuring a single exponential upper bound on the time complexity is highly desirable. For example, Bodleander et al.~\cite{BodlaenderFLPST09} proved that all problems that have finite integer index and satisfy some compactness conditions admit a linear kernel on graphs of bounded genus~\cite{BodlaenderFLPST09}, implying single-exponential running times for such problems.  More recently Cygan et al. developed the ``cut-and-count'' technique to derive (randomized) single-exponential parameterized algorithms for many connectivity problems parameterized by treewidth~\cite{CNP11}. In contrast, some problems are unlikely to have single-exponential algorithms \cite{LokshtanovMS11}.

For \textsc{treewidth-$t$ vertex deletion}, single-exponential parameterized algorithms are known only for $t=0$ and $t=1$. Indeed, for $t=0$ (\textsc{Vertex Cover}), the $O(2^k\cdot n)$-time bounded search tree algorithm is an oft-quoted first example for a parameterized algorithm~\cite{DF99,Nie06,FG06}. For \(t=1\) (\textsc{Feedback Vertex Set}), no single-exponential algorithm was known for many years until Guo \textit{et al.}~\cite{GGH06} and Dehne \textit{et al.}~\cite{DFL07} independently discovered such algorithms. The fastest known deterministic algorithm for this problem runs in time $O(3.83^k\cdot n^2)$~\cite{CCL10}. The fastest known \emph{randomized} algorithm, developed by Cygan et al., runs in $O(3^k\cdot n^{O(1)})$ time~\cite{CNP11}. Very recently, Fomin et al.~\cite{FominLMS11} presented $2^{O(k\log k)}\cdot n^{O(1)}$-time algorithms for \textsc{treewidth-$t$ vertex deletion}. In this paper we prove the following result for $t=2$:

\begin{theorem}\label{th:runtime}
The \KD{} problem can be solved in $2^{O(k)}\cdot n^{O(1)}$ time.
\end{theorem}

\vspace{-0.2cm}
Our single-exponential parameterized algorithm for \KC{} is based on iterative compression. This allows us, with a single-exponential time overhead, to focus on the \emph{disjoint version} of the \KC{} problem: given a solution $S$, find a smaller solution disjoint from $S$. We employ a search tree method to solve the disjoint problem. Although our algorithm shares the spirit of Chen et al.'s search tree algorithm for {\sc Feedback Vertex Set} \cite{ChenFLLV08}, that we want to cover $K_4$-minor instead of $K_3$ requires a nontrivial effort. In order to bound the branching degree by a constant, three key ingredients are exploited. First, we employ protrusion replacement, a technique developed to establish a meta theorem for polynomial-size kernels~\cite{BodlaenderFLPST09,FLST10,FominLMPS11}. We need to modify the existing notions so as to use the protrusion technique in the context of iterative compression. Second, we introduce a notion called the extended SP-decomposition, which makes it easier to explore the structure of treewidth-two graphs. Finally, the technical running time analysis depends on the property of the extended SP-decomposition and a measure which keeps track of the biconnectivity.

\vspace{-0.4cm}
\section{Notation and preliminaries}\label{sec:preliminary}

\vspace{-0.2cm}
We follow standard graph terminology as found in, e.g., Diestel's textbook~\cite{diestel}. Any graph considered in this paper is undirected, loopless and may contain parallel edges. A {\em cut vertex} (resp. {\em cut edge}) is a vertex (resp. an edge) whose deletion strictly increases the number of connected components in the graph. A connected graph without a cut vertex is {\em biconnected}. A subgraph of $G$ is called a {\em block} if it is a maximal biconnected subgraph. A biconnected graph is itself a block. In particular, an edge which is not a part of any cycle is a block as well. For a vertex set $X$ in a graph $G=(V,E)$, the \emph{boundary} $\partial_G(X)$ of $X$ is the set $N(V\setminus X)$, i.e. the set of vertices in $X$ which are adjacent with at least one vertex in $V\setminus X$. We may omit the subscript when it is clear from the context.

\smallskip
\noindent
\textbf{Minors.}
The \emph{contraction} of an edge $e=(u,v)$ in a graph $G$ results in a graph denoted $G/e$ where vertices $u$ and $v$ have been replaced by a single vertex \(uv\) which is adjacent to all the former neighbors of \(u\) and \(v\). A {\em subdivision} of an edge $e$ is the operation of deleting $e$ and introducing a new vertex $x_e$ which is adjacent to both the end vertices of $e$.  A subdivision of a graph $H$ is a graph obtained from $H$ by a series of edge subdivisions. A graph $H$ is a {\em minor} of graph $G$ if it can be obtained from a subgraph of $G$ by contracting edges. A graph $H$ is a {\em topological minor} of $G$ if a subdivision of $H$ is isomorphic to a subgraph $G'$ of $G$. In these cases we say that $G$ contains $H$ as a (topological) minor and that $G'$ is an {\em $H$-subdivision} in $G$. In an $H$-subdivision \(G'\) of $G$, the vertices which correspond to the original vertices of \(H\) are called {\em branching nodes}; the other vertices of \(G'\) are called {\em subdividing nodes}.
It is well known that if the maximum degree of $H$ is at most three, then $G$ contains $H$ as a minor if and only if it contains $H$ as a topological minor~\cite{diestel}. A \emph{$\theta_3$-subdivision} is a graph which consists of three
vertex disjoint paths between two branching vertices.

\smallskip
\noindent
\textbf{Series-parallel graphs and treewidth-two graphs.}
A two-terminal graph is a triple $(G,s,t)$ where $G$ is a graph and the \emph{terminals} $s$, $t$. 
The \emph{series composition} of $(G_1,s_1,t_1)$ and $(G_2,s_2,t_2)$ is obtained by taking the disjoint union of $G_1$ and $G_2$ and identifying $t_1$ with $s_2$. The resulting graph has $s_1$ and $t_2$ as terminals.
The \emph{parallel composition} of $(G_1,s_1,t_1)$ and $(G_2,s_2,t_2)$ is obtained by  taking the disjoint union of $G_1$ and $G_2$ and identifying $s_1$ with $s_2$ and $t_1$ with $t_2$. Series and parallel compositions generalize to any number of two-terminal graphs. \emph{Two-terminal series-parallel graphs} are formed from the single edge and successive series or parallel compositions. A graph $G$ is a \emph{series-parallel graph} (SP-graph) if $(G,s,t)$ is a two-terminal series-parallel graph for some $s,t\in V(G)$.


The recursive construction of a series-parallel graph $G$ defines an \emph{SP-tree} $(T,\mc{X}=\{X_{\alpha}:\alpha \in V(T)\})$, where $T$ is a tree whose leaves correspond to the edges of $G$. Every internal node $\alpha$ is either an \emph{S-node} or a \emph{P-node} and represents the subgraph $G_{\alpha}$ resulting from the series composition or the parallel composition, respectively, of the graphs associated with its children. Every node \(\alpha\) of $T$ is labelled by the set $X_{\alpha}$ of the \emph{terminals} of $G_{\alpha}$. Interested readers are referred to Valdes et al.'s seminal paper on the subject~\cite{VTL82}.
We may assume that an SP-tree satisfies additional conditions. We use, for example, \emph{canonical}\footnote{Full definition, proofs of lemmas, theorems \dots marked by \textbf{$\star$} are also deferred to the appendix} SP-trees for the purpose of analysis, whose definition will not be given in the extended abstract. We remark that any SP-graph can be represented as a canonical SP-tree~\cite{BodlaenderF96} and it can be computed in linear time.




We refer to Diestel's textbook~\cite{diestel} for the definition of the treewidth of a graph $G$ which we denote $tw(G)$.
It is well known that a graph has treewidth at most two graphs if and only if it is $K_4$-minor-free. We also make use of the following alternative characterization: $tw(G)\leqslant 2$ if and only if every block of $G$ is a \emph{series-parallel graph}~\cite{Bodlaender98,BodlaenderF96}.

\smallskip
\noindent
\textbf{Extended SP-decompositon.}
A connected graph $G$ can be decomposed into blocks which are joined by the cut vertices of $G$ in a tree-like manner. To be precise, we can associate a {\em block tree} $\mc{B}_G$ to $G$, in which the node set consists of all blocks and cut vertices of $G$, and a block $B$ and a cut vertex $c$ are adjacent in $\mc{B}_G$ if and only if $B$ contains $c$. To explore the structure of a treewidth-two graph $G$ efficiently, we combine its block tree $\mc{B}_G$ with (canonical) SP-trees of its blocks into an \emph{extended SP-decomposition} as described below. We assume that $G$ is connected: in general, an extended SP-decomposition of $G$ is a collection of extended SP-decompositions of its connected components.

Let $\mc{B}_G$ be the block tree of a treewidth-two graph $G$. We fix an arbitrary cut node $c_{root}$ of $\mc{B}_G$ if one exists. The {\em oriented block tree} $\vec{\mc{B}}_G$ is obtained by orienting the edges of $\mc{B}_G$ outward from $c_{root}$. If $\mc{B}_G$ consists of a single node, it is regarded as an oriented block tree itself.



We construct an extended SP-decomposition of a connected graph $G$ by replacing the nodes of $\vec{\mc{B}}_G$ by the corresponding SP-trees and connecting distinct SP-trees to comply the orientations of edges in $\vec{\mc{B}}_G$. To be precise, an {\em extended SP-decomposition} is a pair $(T, \mc{X}=\{X_{\alpha}:\alpha\in V(T)\})$, where $T$ is a rooted tree whose vertices are called {\em nodes} and $\mc{X}=\{X_{\alpha}:\ \alpha\in V(T)\}$ is a collection of subsets of $V(G)$, one for each node in \(T\). We say that $X_{\alpha}$ is the \emph{label} of node $\alpha$.
\begin{itemize}
\vspace{-0.2cm}
\item For each block $B$ of $G$, let $(T^B,\mc{X}^B)$ be a (canonical) SP-tree of $G[B]$ such that $c(B)$ is one of the terminal associated to the root node of $T^B$. A leaf node of $T^B$ is called an {\em edge node}.
\vspace{-0.2cm}
\item For each cut vertex $c$ of $G$, add to $(T,\mathcal{X})$ a \emph{cut node} $\alpha$ with $X_{\alpha}=\{c\}$.
\vspace{-0.2cm}
\item For each block $B$ of $G$, let the root node of $(T^B,\mathcal{X}^B)$ be a child of the unique
cut node $\alpha$ in \(T\) which satisfies $X_{\alpha}=\{c(B)\}$.
\vspace{-0.2cm}
\item For a cut vertex $c$ of $G$, let $B=B(c)$ be the unique block
  such that $(B,c)\in E(\vec{\mc{B}}_G)$. Let $\beta$ be an
  arbitrary leaf node of the (canonical) SP-tree
  $(T^B,\mathcal{X}^B)$ such that $c\in X_{\beta}$ (note that such
  a node always exists). Make the cut node $\alpha$ of $(T,\mathcal{X})$ labeled by $\{c\}$ a child of the leaf node $\beta$.
\end{itemize}

Let \(\alpha\) be a node of \(T\). Then $T_{\alpha}$ is the subtree of $T$ rooted at node $\alpha$; $E_{\alpha}$ is the set of edges $(u,v)\in E(G)$ such that there exists an edge node $\alpha' \in V(T_{\alpha})$ with $X_{\alpha'}=\{u,v\}$; and  $G_{\alpha}$ is the --- not necessarily induced --- subgraph of $G$ with the vertex set $V_{\alpha}:=\bigcup_{\alpha' \in  V(T_{\alpha})} X_{\alpha'}$ and the edge set $E_{\alpha}$. Recall that \(X_{\alpha}\) is the set of vertices which form the label of the node \(\alpha\), and that \(|X_{\alpha}|\in\{1,2\}\). We define $Y_{\alpha}:=V_{\alpha}\setminus X_{\alpha}$.



Observe that in the construction above, every node $\alpha$ of
$(T,\mathcal{X})$ is either a cut node or corresponds to a node
from the SP-tree $(T^B,\mathcal{X}^B)$ of some block $B$
of $G$. We say that a node $\alpha$ which is not a cut node is
{\em inherited from} $(T^B,\mathcal{X}^B)$, where $B$ is the block
to which \(\alpha\) belongs. Let \(\alpha\) be inherited from
\((T^B,\mathcal{X}^B)\). We use \(T^{B}_{\alpha}\) to denote the
SP-tree naturally associated with the subtree of \(T^{B}\) rooted
at \(\alpha\). By $G^B_{\alpha}$ we
denote the SP-graph represented by the SP-tree $T^{B}_{\alpha}$,
where $(T^B,\mathcal{X}^B)$ inherits $\alpha$. The vertex set of
$G^B_{\alpha}$ is denoted $V^B_{\alpha}$.


We observe that for every node $\alpha$, $G_{\alpha}$ is connected and that $\partial_G(V_{\alpha})\subseteq X_{\alpha}$. It is well-known that one can decide whether $tw(G)\leqslant 2$ in linear time~\cite{VTL82}. It is not difficult to see that in linear time we can also construct an extended SP-decomposition of $G$.

\vspace{-0.4cm}
\section{The algorithm}

\vspace{-0.2cm}
Our algorithm for \KC{} uses various techniques from parameterized complexity. First, an \emph{iterative compression}~\cite{ReedSmithVetta2004} step reduces \KC{} to the so-called \KCC{} problem, where in addition to the input graph we are given a solution set to be improved. Then a {\sc Branch-or-reduce} process develops a \emph{bounded search tree}. 
We start with a definition of the compression problem for \KD{}.


\smallskip
\noindent
\textbf{Iterative compression.} Given a subset $S$ of vertices, a \KC{} $W$ of $G$ is {\em $S$-disjoint} if $W\cap S=\emptyset$. We omit the mention of $S$ when it is clear from the context. If \(\vert W\vert\le k-1\), then we say that $W$ is {\em small}.

\vspace{-0.2cm}
\parfuncdefn{\KCC{} problem}{A graph $G$ and a \KC{} $S$ of $G$}{The integer $k=|S|$}{A small $S$-disjoint \KC{} $W$ of \(G\), if one exists. Otherwise return NO.}

\vspace{-0.7cm}
An FPT algorithm for the \KCC{} problem can be used as a subroutine to solve the \KD{} problem. Such a procedure has now become a standard in the context of iterative compression problems~\cite{ChenFLLV08, JPS11, HHLP11}. 


\begin{lemma}[\textbf{$\star$}] \label{lem:compression}
If \KCC{} can be solved in $c^k\cdot n^{O(1)}$ time, then \KD{} can be solved in $(c+1)^k\cdot n^{O(1)}$ time.
\end{lemma}

\vspace{-0.2cm}
Observe that both $G[V\setminus S]$ and $G[S]$ is $K_4$-minor-free. Indeed if $G[S]$ is not $K_4$-minor-free, then the answer to \KCC{} is trivially NO.

\smallskip
\noindent
\textbf{Protrusion rule.}
A subset $X$ of the vertex set of a graph $G$ is a \emph{$t$-protrusion} of $G$ if $tw(G[X])\leqslant t$ and $|\partial(X)|\leqslant t$. Our algorithm deeply relies on protrusion reduction technique, which made a huge success lately in discovering meta theorems for kernelization~\cite{FLST10,BodlaenderFLPST09}. However, we need to adapt the notions developed for protrusion technique so that we can apply the technique to our ``disjoint'' problem, which arises in the iterative compression-based algorithm.  In essence, our (adapted) protrusion lemma for disjoint parameterized problems says that a 'large' protrusion which is disjoint from the forbidden set $S$ can be replaced by a 'small' protrusion which is again disjoint from $S$. Due to its generality, this result may be of independent interest.

\begin{reduction}[\textbf{$\star$}] \textup{\textbf{(Generic disjoint protrusion rule)}}
\label{rrule:protrusion}
Let $(G,S,k)$ be an instance of \KCC{} and $X$ be a $t$-protrusion such that $X\cap S=\emptyset$. Then there exists a computable function $\gamma(.)$ and an algorithm which computes an equivalent instance in time $O(|X|)$ such that $G[S]$ and $G'[S]$ are isomorphic, $G'-S$ is $K_4$-minor-free, $|V(G')|<|V(G)|$ and $k'\leqslant k$, provided $|X|>\gamma(2t+1)$. 
\end{reduction}

We remark that some of the reduction rules we shall present in the next subsection are instantiations the generic disjoint protrusion rule. However, to ease the algorithm analysis, the generic rule above is used only on $t$-protrusion whose boundary size is 3 or 4. For protrusions with boundary size 1 or 2, we shall instead apply the following explicit reduction rules.

\vspace{-0.4cm}
\subsection{(Explicit) Reduction rules}\label{sec:reduction}

\vspace{-0.2cm}
We say that a reduction rule is \emph{safe} if, given an instance $(G,S,k)$, the rule returns an equivalent instance $(G',S',k')$; that is,  $(G,S,k)$ is a \YES{}-instance if and only if $(G',S',k')$ is. Let $F$ denote the subset $V(G)\setminus S$ of vertices. 
For a vertex $v\in F$, let $N_S(v)$ denote the neighbors of $v$ which belong to $S$. By $N_i\subseteq F$ we refer to the set of vertices $v$ in $F$ with $|N_S(v)|=i$.

The next three rules are simple rule that can be applied in polynomial time. In each of them, $S$ and $k$ are unchanged ($S'=S$, $k'=k$). Observe that reduction rule~\ref{rrule:deg1S} (b) can be seen as a disjoint $1$-protrusion rule.

\begin{reduction}[\textbf{$\star$}]\label{rrule:deg1S}\textup{\textbf{(1-boundary rule)}}
Let $X$ be a subset of $F$. \emph{\textbf{(a)}} If $G[X]$ is a connected component of $G$ or of
$G\setminus e$ for some cut edge $e$, then delete $X$. \emph{\textbf{(b)}} If $|\partial_G(X)|=1$, then delete $X\setminus \partial_G(X)$.
\end{reduction}

\begin{reduction}[\textbf{$\star$}]\label{rrule:deg2} \textup{\textbf{(Bypassing rule)}}
Bypass every vertex $v$ of degree two in $G$ with neighbors $u_1\in V$, $u_2\in F$.
That is, delete $v$ and its incident edges, and add the new edge $(u_1,u_2)$.
\end{reduction}

\begin{reduction}[\textbf{$\star$}]\label{rrule:para}\textup{\textbf{(Parallel rule)}}
If there is more than one edge between $u\in V$ and $v\in F$, then delete all these edges except for one.
\end{reduction}


\vspace{-0.2cm}
The next two reduction rules are somewhat more technical, and their proofs of correctness require a careful analysis of the
structure of the $K_4$-subdivisions in a graph.

\begin{reduction}[\textbf{$\star$}]\label{rrule:consec}\textup{\textbf{(Chandelier rule)}}
  Let $X=\{u_1,\ldots , u_{\ell}\}$ be a subset of $F$, and let
  \(x\) be a vertex in \(S\) such that $G[X]$ contains the path
  $u_1,\ldots , u_{\ell}$, $N_S(u_i)=\{x\}$ for every $i=1,\ldots
  ,\ell$, and vertices $u_{2},\ldots,u_{\ell-1}$ have degree
  exactly $3$ in $G$. If $\ell\geq 4$, contract the edge
  $e=(u_2,u_3)$ (and apply Rule \ref{rrule:para} to remove the
  parallel edges created).
\end{reduction}

The intuition behind the correctness of Chandelier rule~\ref{rrule:consec} is that such a set $X$ cannot host all four branching nodes of a $K_4$-subdivision. Our last reduction rule is an explicit $2$-protrusion rule. In the particular case when the boundary size is exactly two, the candidate protrusions for replacement are either a single edge or a $\theta_3$ (see Figure~\ref{fig:rule-2cut}). 

\begin{reduction}[\textbf{$\star$}]\label{rrule:2cut} \textup{\textbf{(2-boundary rule)}}
  Let $X\subseteq F$ be such that \(G[X]\) is connected,
  $\partial(X)=\{s,t\}$ (and thus, $X\setminus \{s,t\} \subseteq
  N_0$).
  Then we do the following. (1) Delete
  $X\setminus\{s,t\}$. (2) If $G[X]+(s,t)$ is a series parallel graph
  and $|X|>2$, then add the edge \((s,t)\) (if it is not
  present). Else if $G[X]+(s,t)$ is not a series parallel graph
  and $|X|>4$, add two new vertices \(a,b\) and the edges
  $\{(a,b),(a,t),(a,s),(b,t),(b,s)\}$ (see
  Figure~\ref{fig:rule-2cut}).
\end{reduction}

\begin{figure}[htbh]
\begin{center}
\includegraphics[scale=0.7]{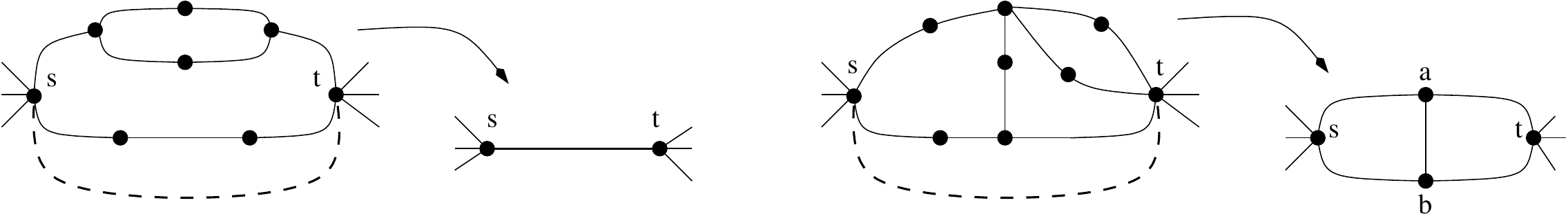}
\end{center}
\vspace{-0.5cm}
\caption{If $G[X]+(s,t)$ is an SP-graph, we can safely replace $G[X]$ by the edge $(s,t)$. Otherwise $G[X]$ can be replaced by a subdivision of $\theta_3$ with poles $a$ and $b$ in which $s$ and $t$ are subdividing nodes.}
\label{fig:rule-2cut}
\end{figure}

An instance of \KCC{} is \emph{reduced} if none of the Reduction rules~\ref{rrule:deg1S} -~\ref{rrule:2cut} applies.

\vspace{-0.4cm}
\subsection{Branching rules}\label{sec:branching}

\vspace{-0.2cm}
A \emph{branching rule} is an algorithm which, given an instance $(G,S,k)$, outputs a set of $d$ instances $(G_1,S_1,k_1)\dots (G_d,S_d,k_d)$ for some constant $d>1$ ($d$ is the branching degree). A branching rule is \emph{safe} if $(G,S,k)$ is a \YES{}-instance if and only if there exists $i$, $1\leqslant i\leqslant d$ such that $(G_i,S_i,k_i)$ is a \YES{} instance. We now present three generic branching rules, with potentially unbounded branching degrees. Later we describe how to apply these rules so as to bound the branching degree by a constant. Given a vertex $s\in S$, we denote by $cc_S(s)$ the connected component of $G[S]$ which contains $s$. Likewise, $bc_S(s)$ denotes the biconnected component of $G[S]$ containing $s$. It is easy to see that three branching rules below are safe.

\begin{branching} \label{rule:br1}
Let $(G,S,k)$ be an instance of \KCC{} and let $X$ be a subset of $F$ such that 
$G[S\cup X]$ contains a $K_4$-subdivision. Then branch into the instances $(G-\{x\},S,k-1)$ for every $x\in X$.
\end{branching}

\begin{branching} \label{rule:br2}
Let $(G,S,k)$ be an instance of \KCC{} and let $X$ be be a connected subset of $F$. If $S$ contains two vertices $s_1$ and $s_2$ each having a neighbor in $X$ and such that $cc_S(s_1)\neq cc_S(s_2)$, then branch into the instances
\begin{itemize}
\vspace{-0.3cm}
\item $(G-\{x\},S,k-1)$ for every $x\in X$
\vspace{-0.2cm}
\item $(G,S\cup X, k)$
\end{itemize}
\end{branching}


\begin{branching} \label{rule:br3} 
Let $(G,S,k)$ be an instance of
  \KCC{} and let $X$ be a connected subset of $F$. If $S$ contains
  two vertices $s_1$ and $s_2$ each having a neighbor in $X$ such
  that $cc_S(s_1)=cc_S(s_2)$ and $bc_S(s_1)\neq bc_S(s_2)$, then
  branch into the instances
\begin{itemize}
\vspace{-0.3cm}
\item $(G-\{x\},S,k-1)$ for every $x\in X$
\vspace{-0.2cm}
\item $(G,S\cup X, k)$
\end{itemize}
\end{branching}

We shall apply branching rule \ref{rule:br1} under three different situations: (i) $X$ is a singleton $\{x\}$ for every $x\in F$, (ii) $X$ is connected, and (iii) $X$ consists of a pair of non-adjacent vertices of $F$. Let us discuss these three settings in further details. An instance $(G,S,k)$ is said to be a \emph{simplified instance} if it is a reduced instance and if none of the branching rules~\ref{rule:br1}~-~\ref{rule:br3} applies on singleton sets $X=\{v\}$, for any $v\in F$.
A simplified instance, in which branching rule \ref{rule:br1} cannot be applied under (i), has a useful property.

\begin{lemma}[\textbf{$\star$}] \label{cor:simplified-deg2}
If $(G,S,k)$ is a simplified instance of \KCC{}, then $F=N_0\cup N_1\cup N_2$.
\end{lemma}

An instance $(G,S,k)$ of \KCC{} is \emph{independent} if (a) $F$ is an independent set; (b) every vertex of $F$ belongs to $N_2$; (c) the two neighbors of every vertex of $F$ belong to the same biconnected component of $G[S]$ and (d) $G[S\cup \{x\}]$ is $K_4$-minor-free for every $x\in F$. In essence, next lemma shows that the instance is independent once branching rule \ref{rule:br1} has been exhaustively applied under (ii).

\begin{theorem}[\textbf{$\star$}] \label{th:independent}
Let $(G,S,k)$ be an instance of \KCC{}. If none of the reduction rules applies nor branching rules on connected subsets $X\subseteq F$ applies, then $(G,S,k)$ is an independent instance.
\end{theorem}

\vspace{-0.2cm}
Next lemma shows that in an independent instance, it is enough to cover the $K_4$-subdivisions containing exactly two vertices of $F$. To see this, we construct an auxiliary graph $G^*(S)$ as follows: its vertex set is $F$; $(u,v)$ is an edge in $G^*(S)$ if and only if $G[S\cup\{u,v\}]$ contains $K_4$ as a minor. Then the following theorem holds, which essentially states that we obtain a solution for \KCC{} by applying branching rule \ref{rule:br1} exhaustively under (iii).


\begin{theorem}[\textbf{$\star$}] \label{lem:transVC}
Let $(G,S,k)$ be an independent instance of \KCC{}. Then $W\subseteq F$ is a disjoint \KC{} of $G$ if and only if it is a vertex cover of $G^*(S)$.
\end{theorem}

Observe that we do not need to build $G^*(S)$ to solve the \KCC{} problem on an independent instance\footnote{A more careful analysis shows that $G^*(S)$ is a circle graph. As {\sc Vertex Cover} is polynomial time solvable on circle graphs, so is \KCC{} problem on an independent instance.}. Indeed, for every pair of vertices $u,v\in F$, it is enough to test whether $G[S\cup\{u,v\}]$ contains $K_4$ as a minor (this can be done in linear time~\cite{VTL82}) and if so we apply branching rule~\ref{rule:br1} on the set $X=\{u,v\}$.

\subsection{Algorithm and complexity analysis }\label{sec:analysis}

\vspace{-0.2cm}
Let us present the whole search tree algorithm. At each node of the computation tree associated with a given instance $(G,S,k)$, one of the followings operations is performed. As each operation either returns a solution (as in (a),(e)) or generates a set of instances (as in (b)-(d)), the overall application of the operations can be depicted as a search tree.

\begin{enumerate}
\vspace{-0.3cm}
\item[(a)] if ($k<0$) or ($k\leq 0$, $tw(G)> 2$) or ($tw(G[S]) > 2$), then return no; 
\vspace{-0.3cm}
\item[(b)] if the instance is not reduced, apply one of Reduction rules \ref{rrule:deg1S}--\ref{rrule:2cut} (note that we apply Reduction rules \ref{rrule:deg1S}--\ref{rrule:consec} first whenever
possible, and Reduction rule \ref{rrule:2cut} is applied when none of the rules \ref{rrule:deg1S}--\ref{rrule:consec} can be applied); 
\vspace{-0.3cm}
\item[(c)] if the instance is not simplified, apply one of Branching rules \ref{rule:br1}--\ref{rule:br3} on the singleton
sets $\{x\}$ for each $x\in F$; 
\vspace{-0.3cm}
\item[(d)] if the instance is simplified, apply the procedure {\sc Branch-or-reduce}; 
\vspace{-0.3cm}
\item[(e)] if the application of {\sc Branch-or-reduce} marks every node of $(T,\mc{X})$, the instance is an independent instance; solve it in $2^k\cdot n^{O(1)}$ using branching rule~\ref{rule:br1} on pairs of vertices of $F$.
\end{enumerate}

\vspace{-0.2cm}
We now describe the procedure {\sc Branch-or-reduce}
as a systematic way of applying the branching and reduction
rules. It works in a bottom-up manner on an extended
SP-decomposition $(T,\mathcal{X})$ of $G[F]$. Initially the nodes
of $(T,\mathcal{X})$ are unmarked. Starting from a lowest node,
{\sc Branch-or-reduce} recursively tests if we can apply one of
the branching rules on a subgraph associated with a lowest unmarked node. If the
branching rules do not apply, it may be due to a large
protrusion. In that case, we detect the protrusion (see Lemma
\ref{lem:2attach}) and reduce the instance using the protrusion
 rule \ref{rrule:protrusion}. Once either a branching
rule or the protrusion rule has been applied, the procedure {\sc
  Branch-or-reduce} terminates. The output is a set of instances
of \KCC{}, possibly a singleton. 

\begin{algorithm}[h]
  \KwIn{A simplified instance $(G,S,k)$ of \KCC{}, together with
    an extended SP-decomposition $(T,\mathcal{X})$ of
    $G[F]$.}
  \KwOut{A set of instances of \KCC{}.}
  \BlankLine

\While{\(T\) contains unmarked nodes}{
  \lnl{l:alpha}
  Let \(\alpha\) be an unmarked node at the farthest  distance from the root of \(T\)\;
	\lnl{l:br2-test}
	\If{$S$ contains two vertices $x_u\in N_S(u)$ and $x_v\in N_S(v)$ with $u,v\in V_{\alpha}$ and $cc_S(x_u)\neq cc_S(x_v)$}{
		\lnl{l:Xbr2}
		Let $X$ be a path in $G_{\alpha}$ between two such vertices $u$ and $v$ such that $X\setminus\{u,v\}\subseteq N_0$\;
		\lnl{l:br2}		
		Apply Branching rule~\ref{rule:br2} to \(X\); terminate\;
		}
	\lnl{l:br3-test}
	\If{$S$ contains two vertices $x_u\in N_S(u)$ and $x_v\in N_S(v)$ with $u,v\in V_{\alpha}$ and $bc_S(x_u)\neq bc_S(x_v)$}{
		\lnl{l:Xbr3}	
		Let $X$ be a path in $G_{\alpha}$ between two such vertices $u$ and $v$ such that $X\setminus\{u,v\}\subseteq N_0$\;
		\lnl{l:br3}
		Apply Branching rule~\ref{rule:br3} to \(X\); terminate\;
		}
	\lnl{l:br1-test}
	\If{$G[S\cup V_{\alpha}]$ contains a $K_4$-subdivision}{
		\lnl{l:Xbr1}
		Let $X \subseteq V_{\alpha}$ be a connected set such that $G[S+X]$ contains a $K_4$-subdivision\;
		\lnl{l:br1}
		Apply Branching rule~\ref{rule:br1} to \(X\); terminate\;
		}
	\lnl{l:protrusion-test}
	\If{$\alpha$ is a P-node and $|V^B_{\alpha}|\geqslant \gamma(9)$}{
		\lnl{l:Xprot}
		$X = V^B_{\alpha}$ is a 4-protrusion (see Lemma~\ref{lem:2attach})\;
		\lnl{l:prot}
		Apply the protrusion Reduction rule~\ref{rrule:protrusion} with $X$; terminate\;
		}
    \lnl{l:mark}	
       	Mark the node $\alpha$\;
}
\caption{{\sc Branch-or-reduce}} \label{K4algo}
\end{algorithm}

The complexity analysis relies on a series of technical lemmas such as Lemma \ref{lem:2attach}. We say that a path $P$ avoids a set $X$ if no internal vertex of $P$ belongs to $X$. To simplify the notation, we use $G_{\alpha}$ instead of $G[F]_{\alpha}$ for a node $\alpha$ of \(T\). Similarly, we use the names $V_{\alpha}$, $Y_{\alpha}=V_{\alpha}\setminus X_{\alpha}$ and $V^B_{\alpha}$ to denote the various named subsets of \(V(G[F]_{\alpha})\).
%
\begin{lemma}[\textbf{$\star$}]\label{lem:K4-3+1} Let $W$ and $Z$ be disjoint
  vertex subsets of a graph $G$ such that $G[W]$
  is biconnected, $G[Z]$ is connected and $|N_W(Z)|\geq 3$. Then
  $G[W\cup Z]$ contains a $K_4$-subdivision.
\end{lemma}

\begin{lemma} \label{lem:2attach}
Let $(G,S,k)$ be a simplified
  instance and let $\alpha$ be a lowest node of the
  extended SP-decomposition $(T,\mathcal{X})$ of $G[F]$ which is considered at line~\ref{l:protrusion-test} of Algorithm~\ref{K4algo}. If $\alpha$ is a
  P-node inherited from the SP-tree of block $B$, then $|\partial_G(V^B_{\alpha})\setminus X_{\alpha}|\leq 2$ and $V^B_{\alpha}$ is a $4$-protrusion.

\end{lemma}
\begin{proof}



\vspace{-0.3cm}
  As $\alpha$ is a P-node, $G^B_{\alpha}$ is biconnected. We argue $|\partial_G(V^B_{\alpha})\setminus X_{\alpha}|\leq 2$ and the second statement easily follows.
    Suppose $\partial_G(V^B_{\alpha})\setminus X_{\alpha}$ contains three distinct vertices, say, $x$, $y$ and
  $z$. We claim that there exist three
  internally vertex-disjoint paths $P_x$, $P_y$ and $P_z$ from $S$
  to each of $x$, $y$ and $z$ avoiding $V^B_{\alpha}$. Without loss of generality, we show that $G[S\cup V_{\alpha}]$ contains a path $P_x$ between $S$ and $x$ avoiding $V^B_{\alpha}$ and the claim follows as a corollary. If $x\in N_1\cup N_2$, then it is trivial. Suppose $x \notin N_1\cup N_2$ and thus $x$ is a cut vertex of \(G[F]\). Then $(T,\mathcal{X})$ contains a cut node $\beta$ with $X_{\beta}=\{x\}$ such that $\beta$ is a descendent of $\alpha$. It can be shown\footnote{Lemma~\ref{lem:withattach} in the appendix} that $Y_{\beta}\cap (N_1\cup N_2)\neq   \emptyset$.  Since $G_{\beta}$ is connected, $G[S\cup V_{\beta}]$ contains a path \(P_x\) between $S$ and $x$ and $P_x$ is a path avoiding $V^B_{\alpha}$.

  As $\alpha$ fails the test of line~\ref{l:br2-test}, the vertices
  of $N_{S}(V_{\alpha})$ belong to the same connected component,
  say $C$, of $G[S]$. Now Lemma~\ref{lem:K4-3+1} applies to the
  biconnected graph $G^B_{\alpha}$ and $(C\cup P_x\cup P_y\cup
  P_z)\setminus \{x,y,z\}$, showing that $G[V^B_{\alpha}\cup P_x
  \cup P_y \cup P_z \cup S]$ contains a $K_4$-subdivision: a
  contradiction to the fact that Branching rule~\ref{rule:br1}
  does not apply. Therefore, $\partial_G(V^B_{\alpha})\setminus X_{\alpha}$ contains at most two vertices.
\end{proof}

The next two lemmas show that applying {\sc Branch-or-reduce} in a bottom-up manner enables us to bound the branching degree of the {\sc Branch-or-reduce} procedure. Lemma \ref{lem:marked-cst-size} states that for every marked node $\alpha$, the graph $G_{\alpha}$ is of constant-size.

\begin{lemma}[\textbf{$\star$}] \label{lem:marked-cst-size}
Let $(G,S,k)$ be a simplified instance of \KCC{} and let $\alpha$ be a marked node of the extended SP-decomposition $(\mathcal{X},T)$ of $G[F]$. Then $|V_{\alpha}|\leqslant c_1:=12(\gamma(8)+2c_0)$.
\end{lemma}

\begin{lemma}[\textbf{$\star$}] \label{lem:poly-branching} Let $(G,S,k)$ be a
  simplified instance of \KCC{} and let $\alpha$ be a lowest
  unmarked node of $(T,\mathcal{X})$
  of $G[F]$.  In polynomial time, one can find
\begin{enumerate}
\vspace{-0.2cm}
\item[(a)] a path $X$ of size at most $2c_1$ satisfying the conditions
  of line~\ref{l:Xbr2} (resp. line~\ref{l:Xbr3}) if the test at
  line~\ref{l:br2-test} (resp.~\ref{l:br3-test}) succeeds;
\vspace{-0.2cm}
\item[(b)] a subset $X\subseteq V_{\alpha}$ of size bounded by
  $2c_1$ satisfying the condition of line~\ref{l:Xbr1} if the test at
  line~\ref{l:br1-test} succeeds;
\end{enumerate}
\end{lemma}

\vspace{-0.2cm}
\noindent For running time analysis of our algorithm, we introduce the following measure

\centerline{$\mu := (2c_1+2)k +(2c_1+2)\# cc(G[S]) + \#bc(G[S])$}

\smallskip
\noindent
where $\#cc(G[S])$ and $\#bc(G[S])$ respectively denote the number
of connected and biconnected components of $G[S]$.

\begin{reminder}{Theorem~\ref{th:runtime}}
The \KD{} problem can be solved in $2^{O(k)}\cdot n^{O(1)}$ time.
\end{reminder}
\vspace{-0.2cm}
\begin{proof}

  Due to Lemma \ref{lem:compression}, it is sufficient to show
  that one can solve \KCC{} in time $2^{O(k)}\cdot n^{O(1)}$. The
  recursive application of operations (a)-(e) at the beginning of the section to a given instance $(G,S,k)$ produces a search tree
  $\Upsilon$. It is not difficult to see that $(G,S,k)$ is a
  \YES{}-instance if and only if at least one of the leaf nodes in
  $\Upsilon$ corresponds to a \YES{}-instance. This follows from
  the fact that reduction and branching rules are safe.

 Let us see the running time to apply the operations (a)-(e) at each node of $\Upsilon$. Every instance corresponding to a leaf node either is a trivial instance or is an independent
  instance (see Theorem~\ref{th:independent}) which can be solved in
 $2^k\cdot n^{O(1)}$ using branching rule~\ref{rule:br1} on pairs of vertices of $F$ (see Theorem~\ref{lem:transVC}). Clearly, the
  operations (a)--(c) can be applied in polynomial time. Consider
  the operation (d). The while-loop in the algorithm {\sc
    Branch-or-reduce} iterates $O(n)$ times. At each
  iteration, we are in one of the three situations: we detect in polynomial time (Lemma
  \ref{lem:poly-branching}) a connected subset $X$ on which to
  apply one of Branching rules, or apply the protrusion rule in
  polynomial time (Reduction rule~\ref{rrule:protrusion}), or
  none of these two cases occur and the node under consideration is marked.

Observe that the branching degree of the search tree is at most $2c_1+1$ by Lemma \ref{lem:poly-branching}. To bound the size of $\Upsilon$, we need the following claim.

\begin{claim}\label{cl:measure}
\emph{In any application of Branching rules \ref{rule:br1}--\ref{rule:br3}, the measure $\mu$ strictly decreases.}
\end{claim}
\vspace{-0.2cm}
\begin{proofof}
  The statement holds for Branching rule \ref{rule:br1} since
  \(k\) reduces by one and \(G[S]\) is unchanged. Recall that
  Branching rules \ref{rule:br2} and \ref{rule:br3} put a
  vertex in the potential solution or add a path $X\subseteq F$ to $S$. In the first
  case, $\mu$ strictly decreases because \(k\) decreases and
  \(\#cc(G[S])\) and \(\#bc(G[S])\) remain unchanged. Let us see that $\mu$ strictly decreases also when we add
  a path $X$ to $S$.

  If Branching rule \ref{rule:br2} is applied, the number of
  biconnected components may increase by at most $2c_1+1$. This
  happens if every edge on the path $X$ together with the two
  edges connecting the two end vertices of $X$ to $S$ add to the
  biconnected components of $G[S\cup X]$. Hence we have that the
  new value of \(\mu\) is $\mu'= (2c_1+2)k+(2c_1+2)\#cc(G[S\cup
  X])+\#bc(G[S\cup X])\leqslant
  (2c_1+2)k+(2c_1+2)(\#cc(G[S])-1)+(\#bc(G[S])+2c_1+1)\leqslant
  \mu-1$. It remains to observe that an application Branching rule \ref{rule:br3}  strictly decreases the number of biconnected components while does not increase the number of connected components. Thereby $\mu'\leqslant \mu-1$.
\end{proofof}

By Claim \ref{cl:measure}, at every root-leaf computation path in $\Upsilon$ we have at most $\mu = (2c_1+2)k +(2c_1+2)\#cc(G[S]) + \#bc(G[S]) \leq (4c_1+5)k$ nodes at which a branching rule is applied. Since we branch into at most $(2c_1+1)$ ways, the number of leaves is bounded by $(2c_1+1)^{(4c_1+5)k}$. Also note that any root-leaf computation path contains $O(n)$ nodes at which a reduction rule is applied since any reduction rule strictly decreases the size of the instance and does not affect $G[S]$. It follows that the running time is bounded by $((4c_1+5)k+ O(n))\cdot (2c_1+1)^{(4c_1+3)k} \cdot poly(n)=2^{O(k)}\cdot n^{O(1)}$.
\end{proof}

\vspace{-0.4cm}
\section{Conclusion and open problems}

\vspace{-0.3cm}
Due to the use of the generic protrusion rule (on $t$-protrusion for $t=3$ or $4$), the result in this paper is existential. A tedious case by case analysis would eventually leads to an explicit $c^k\cdot n^{O(1)}$ exponential FPT algorithm for some constant value $c$. It is an intriguing challenge to reduce the basis to a small $c$ and/or get a simple proof of such an explicit algorithm. More generally, it would be interesting to investigate the systematic instantiation of protrusion rules.

We strongly believe that our method will apply to similar
problems. The first concrete example is the parameterized
\textsc{Outerplanar Vertex Deletion}, or equivalently the
\textsc{$\{K_{2,3},K_4\}$-minor cover} problem. For that problem, we need to
adapt the reduction and branching rules in order to preserve
(respectively, eliminate) the
existence of a $K_{2,3}$ as well. For example, the by-passing rule
(Reduction rule~\ref{rrule:deg2}) may destroy a $K_{2,3}$ unless
we only bypass a degree-two vertices when it is adjacent to
another degree-two vertex. Similarly in Reduction
Rule~\ref{rrule:2cut}, we cannot afford to replace the set $X$ by
an edge. It would be safe with respect to $\{K_{2,3},K_4\}$-minor if
instead $X$ is replaced by a length-two path or by two parallel
paths of length two (depending on the structure of $X$). So we
conjecture that for \textsc{Outerplanar Vertex Deletion} our
reduction and branching rules can be adapted to design a single
exponential FPT algorithm.

A more challenging problem would be to get a single exponential FPT algorithm for the \textsc{treewidth-$t$ vertex deletion} for any value of $t$. Up to now and to the best of our knowledge, the fastest algorithm runs in $2^{O(k\log k)}\cdot n^{O(1)}$~\cite{FominLMS11}.

\smallskip
\noindent
\textbf{Acknowledgements.}
We would like to thank Saket Saurabh for his insightful comments on an early draft and Stefan Szeider for pointing out the application of our problem in Bayesian Networks.


{\small

}
\newpage
\appendix

\section{Definitions}

\subsection{Minors and tree-width}

\begin{figure}[htbh]
\begin{center}
\includegraphics[scale=0.75]{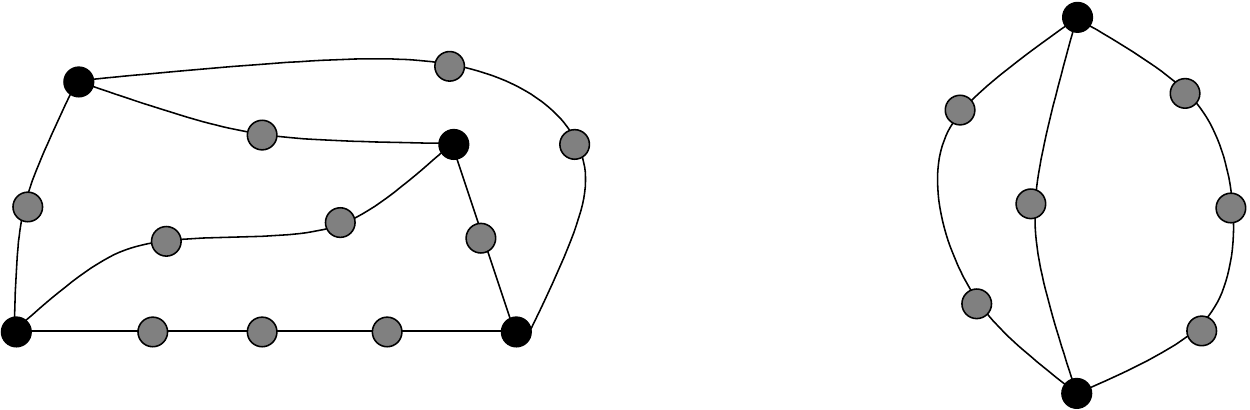}
\end{center}
\caption{A $K_4$-subdivision on the left and a
  $\theta_3$-subdivision on the right. The black vertices are the
  branching nodes. }
\label{fig:k4-subdivision}
\end{figure}

\begin{observation}\label{nocut}
  A $K_4$-subdivision is biconnected; equivalently,
  it is connected and does not contain a cut vertex.
\end{observation}

Since there are three distinct paths between any two branching
nodes in a $K_4$-subdivision, we need at least three vertices in
order to separate any two of them. Hence we have:

\begin{observation}\label{cross2cut}
  Let $\{s,t\}$ be a separator of graph $G$, and let $H$ be a
  $K_4$-subdivision in $G$. Then there exists a connected
  component $X_0$ of $G-\{s,t\}$ such that all four branching
  nodes of $H$ belong to $X_0\cup \{s,t\}$.
\end{observation}

A {\em tree decomposition} of $G$ is a pair $(T,\mc{X})$, where
$T$ is a tree whose vertices we will call {\em nodes} and
$\mc{X}=\{X_{i}:i\in V(T)\}$ is a collection of subsets of $V(G)$
(called {\em bags}) with the following properties:
\begin{enumerate}
\item $\bigcup_{i \in V(T)} X_{i} = V(G)$,

\item for each edge $(v,w) \in E(G)$, there is an $i\in V(T)$
such that $v,w\in X_{i}$, and

\item for each $v\in V(G)$ the set of nodes $\{ i :\ v \in X_{i}
\}$ form a subtree of $T$.
\end{enumerate}

The {\em width} of a tree decomposition $(T,\{ X_{i}:\ i \in V(T)
\})$ equals $\max_{i \in V(T)} \{|X_{i}| - 1\}$. The {\em
  treewidth} of a graph $G$ is the minimum width over all tree
decompositions of $G$. We use the notation $tw(G)$ to denote the
treewidth of a graph $G$.

\subsection{Block, canonical SP-tree and extended SP-decomposition}

Without loss of generality, we may assume~\cite{BodlaenderF96}
that an SP-tree satisfies the following conditions: (1) an S-node
does not have another S-node as a child; each child of an S-node
is either a P-node or a leaf; and (2) a P-node has exactly two
children --- see Figure~\ref{fig:sp-tree}.

\begin{figure}[htbh]
\begin{center}
\includegraphics[scale=0.8]{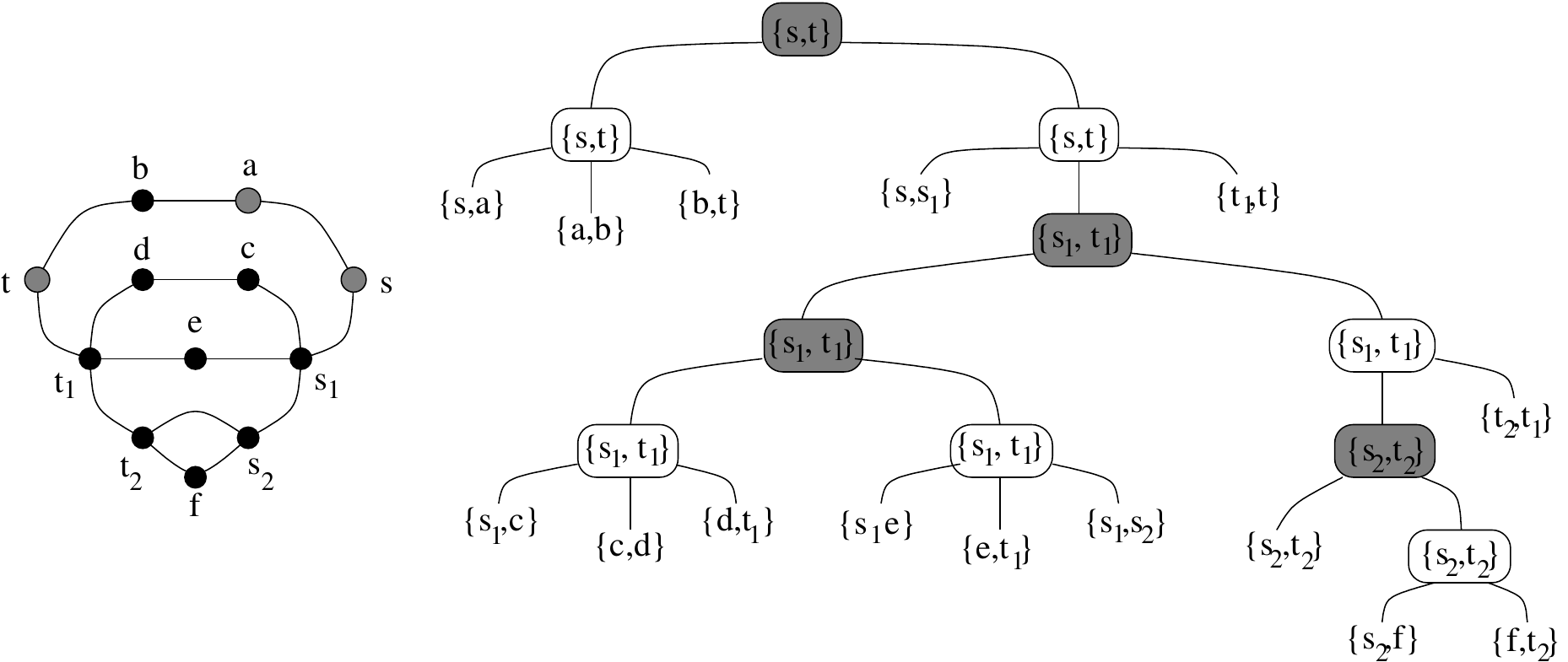}
\end{center}
\caption{A canonical SP-tree.  P-nodes are coloured grey and
  S-nodes are coloured white. Observe that as P-nodes are binary
  and may have a P-node as a child, while S-nodes do not have
  any S-node as a child, conditions (1) and (2) are satisfied.}
\label{fig:sp-tree}
\end{figure}

By Lemma~\ref{biconnectedSP}, we may further assume that for a biconnected series-parallel graph \(G\) and any fixed vertex $s\in V(G)$, (3) \(G\) has an SP-tree
whose root is a P-node with $s$ as one of its two terminals. We say that an SP-tree is {\em canonical} if it
satisfies the conditions (1) and (2), and also (3) when $G$ is
biconnected.

\begin{lemma}\label{biconnectedSP}\textup{\textbf{\cite{Eppstein92}}}
  Let $G$ be a series-parallel graph, and let \(s,t\) be two
  vertices in \(G\). Then $G$ is an SP-graph with terminals $s$
  and $t$ if and only if $G+(s,t)$ is an SP-graph. Moreover, if
  \(G\) is biconnected, then the last operation is a parallel
  join.
\end{lemma}


The following is a well-known characterization relating forbidden minors, treewidth, and series-parallel graphs~\cite{Bodlaender98,BodlaenderF96}.
\begin{lemma}\label{SPcharac}
Given a graph $G$, the followings are equivalent.
\begin{itemize}
\vspace{-0.2cm}
\item $G$ does not contain $K_4$ as a minor (That is, \(G\) is
  $K_4$-minor-free.).
\vspace{-0.2cm}
\item The treewidth of $G$ is at most two.
\vspace{-0.2cm}
\item Every block of $G$ is a series-parallel graph.
\end{itemize}
\end{lemma}

It is well-known that one can decide whether $tw(G)\leqslant 2$ in linear time~\cite{VTL82}. It is not difficult to see that in linear time we can also construct an extended SP-decomposition of $G$. Though the next lemma is straightforward, we sketch the proof for completeness.

\begin{lemma}\label{testSP}
  Given a graph $G$, one can decide whether $tw(G)\leqslant 2$ (or
  equivalently, whether $G$ is $K_4$-minor-free) in linear time. Further, we can
  construct an extended SP-decomposition of \(G\) in linear time
  if $tw(G)\leqslant 2$.
\end{lemma}
\begin{proof}
  The classical algorithm due to Hopcroft and
  Tarjan~\cite{HopcroftT73} identifies the blocks and cut vertices
  of $G$ in linear time. Due to Lemma~\ref{SPcharac}, testing
  $tw(G)\leq 2$ reduces to testing whether each block of $G$ is a
  series-parallel graph. It is known~\cite{VTL82} that the
  recognition of a series-parallel graph and the construction of
  an SP-decomposition can be done in linear time. Further, an
  SP-decomposition can be transformed into a canonical
  SP-decomposition in linear time. Given an oriented block tree
  $\vec{\mc{B}}_G$ and a canonical SP-decomposition for every
  block, we can construct the extended SP-decomposition in linear
  time, and the statement follows.
\end{proof}

\begin{figure}[htbh]
\begin{center}
\includegraphics[scale=0.75]{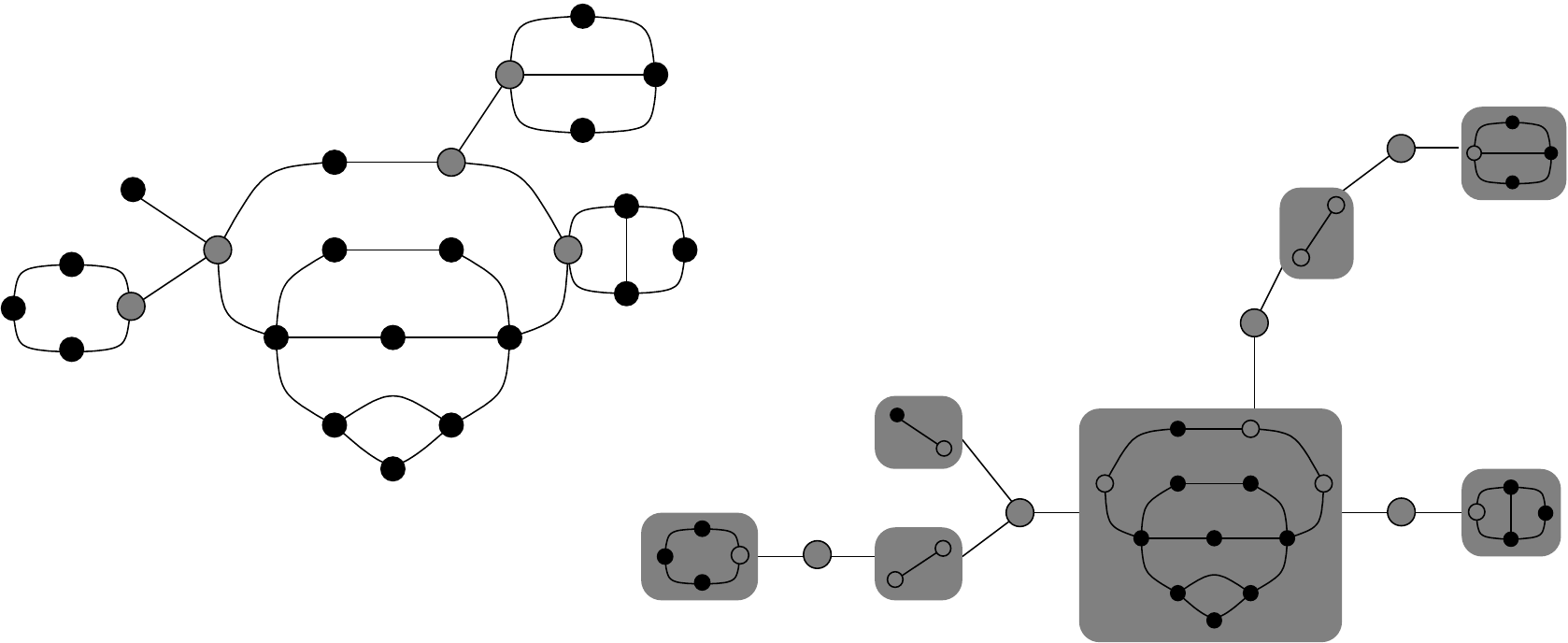}
\end{center}
\caption{A $K_4$-minor-free graph $G$ and its block tree $\mc{B}_G$.}
\label{fig:block-tree}
\end{figure}

\section{Proof of Generic disjoint protrusion rule}



\begin{definition}[$t$-Boundaried Graphs]
A \emph{$t$-boundaried graph} is a graph $G=(V,E)$ with $t$ distinguished vertices, uniquely labeled from 1 to $t$. The set $\partial(G)\subseteq V$ of labeled vertices is called the \emph{boundary} of $G$. The vertices in $\partial(G)$ are referred to as \emph{boundary vertices} or \emph{terminals}.
\end{definition}

\begin{definition} [Gluing by $\oplus$] Let $G_1$ and $G_2$ be two
  $t$-boundaried graphs. We denote by $G_1\oplus G_2$ the
  $t$-boundaried graph such that: its vertex set is obtained by
  taking the disjoint union of $V(G_1)$ and $V(G_2)$, and
  identifying each vertex of $\partial(G_1)$ with the vertex of
  $\partial(G_2)$ having the same label; and its edge set is the
  union of $E(G_1)$ and $E(G_2)$. (That is, we glue $G_1$ and
  $G_2$ together on their boundaries.)
\end{definition}

Many graph optimization problems can be rephrased as a task of
finding an optimal number of vertices or edges satisfying a
property expressible in Monadic Second Order logic (MSO). A
parameterized graph problem $\Pi \subseteq \Sigma^* \times
\mathbb{N}$ is given with a graph $G$ and an integer $k$ as an
input. When the goal is to decide whether there exists a subset
$W$ of at most $k$ vertices for which an MSO-expressible property
$P_{\Pi}(G,W)$ holds, we say that $\Pi$ is a \mMSO ~graph
problem. When $P_{\Pi}(G,\emptyset)$ holds, we write that $P_{\Pi}(G)$ holds (or that $G$ satisfies $P_{\Pi}$). In the (parameterized) \emph{disjoint version} $\Pi^d$ of
a \mMSO ~problem $\Pi$, we are given a triple $(G,S,k)$, where $G$
is a graph, $S$ a subset of $V(G)$ and $k$ the parameter, and we
seek for a solution set $W$ which is disjoint from $S$, and whose
size is at most $k$. The fact that a set $W$ is such a solution is
expressed by the MSO-property $P_{\Pi^d}(G,S,W):P_{\Pi}(G,W) \wedge
(S\cap W = \emptyset)$.


\begin{definition}
  For a disjoint parameterized problem $\Pi^d$ and two
  $t$-boundaried graphs~\footnote{We use this notation
    since later in this section, $G_p$ plays the role of a (large)
    \textbf{p}rotrusion and $G_r$, its \textbf{r}eplacement.} $G_p$ and $G_r$,
  we say that $G_p\equiv_{\Pi^d} G_r$ if there exists a constant
  $c$ such that for all $t$-boundaried graphs $G$, for every vertex set
  $S\subseteq V(G)\setminus \partial(G)$, and for every integer
  $k$,
\begin{center}
$(G_p\oplus G, S,k)\in \Pi^d$ if and only if $(G_r\oplus G,S,k+c)\in \Pi^d$
\end{center}

\end{definition}

\begin{definition} [Disjoint Finite integer index] \label{def:FII}
  For a disjoint parameterized graph problem $\Pi^d$, we say that
  $\Pi^d$ has {\em disjoint finite integer index} if the following
  property is satisfied: for every $t$, there exists a finite set
  $\mc{R}$ of $t$-boundaried graphs such that for every
  $t$-boundaried graph $G_p$ there exists $G_r\in \mc{R}$ with
  $G_p\equiv_{\Pi^d} G_r$. Such a set $\mc{R}$ is called a
  \emph{set of representatives} for $(\Pi^d,t)$.
\end{definition}

It is often convenient to pair up
a $t$-boundaried graph $G$ with a set $W\subseteq V(G)$ of
vertices. We define $\mc{H}_t$ to be the set of pairs $(G,W)$,
where $G$ is a $t$-boundaried graph and $W\subseteq V(G)$.  For an
\mMSO{} problem $\Pi$ and a $t$-boundaried graph $G$, we define
the {\em signature function} $\zeta_G:\mc{H}_t\rightarrow
\mathbb{N}\cup \{\infty\}$ as follows.
\begin{equation*}
\zeta_G((G',W'))=\left \{
\begin{array}{ll}
\infty & \text{ if } \nexists W \subseteq V(G) ~s.t.~ P_{\Pi}(G\oplus G',W\cup W') \\
\min_{W\subseteq V(G)} \{|W| : P_{\Pi}(G\oplus G',W\cup W')\} & \text{ otherwise}
\end{array} \right.
\end{equation*}
To ease the notation, we write $\zeta_G(G',W')$ to denote $\zeta_G((G',W'))$.

\begin{definition} [Strong monotonicity]
  A \mMSO{} problem $\Pi$ is said to be {\em strongly monotone} if
  there exists a function $f: \mathbb{N} \rightarrow \mathbb{N}$
  satisfying the following condition: for every $t$-boundaried
  graph $G$, there exists a set $W_G \subseteq V(G)$ such that for
  every $(G',W')\in \mc{H}_t$ with finite value $\zeta_G(G',W')$,
  $P_{\Pi}(G\oplus G',W_G\cup W')$ holds and $|W_G|\leqslant
  \zeta_G(G',W')+f(t)$.
\end{definition}

Bodlaender et al. show~\cite[proof of Lemma 13]{BodlaenderFLPST09}
that if $\mathcal{F}$ is a finite set of connected planar graphs,
then \textsc{$\mathcal{F}$-minor cover} problem is strongly
monotone. The following lemma is a corollary of this fact. We give
the proof for completeness.

\begin{lemma}
The $K_4$-\textsc{minor cover} problem is strongly monotone.
\end{lemma}
\begin{proof}
Let $G$ be a $t$-boundaried graph and $\partial(G)$ be its boundary. Let $W\subseteq V(G)$ be a minimum size vertex subset such that $G[V\setminus W]$ is $K_4$-minor-free. Define $W_G=W\cup\partial(G)$. Then for every pair $(G',W')\in\mathcal{H}_t$ such that $\zeta_G(G',W')$ is finite, $W_G\cup W'$ is a \KC{} of $G\oplus G'$ and moreover by construction $|W_G|\leqslant \zeta_G(G',W')+t$.
\end{proof}

\begin{lemma}\label{lem:disjointFII}
Let $\Pi$ be a strongly monotone \mMSO{} problem. Then its disjoint version $\Pi^d$ has disjoint finite integer index.
\end{lemma}
\begin{proof}
We consider the following equivalence relation $\sim_{\Pi}$ on $\mc{H}_t$:
$(G,W)\sim_{\Pi} (G',W')$ if and only if for every $(G_p,W_p)\in\mc{H}_t$ we have
$$P_{\Pi} (G_p\oplus G, W_p \cup W) \Leftrightarrow P_{\Pi}(G_p\oplus G', W_p\cup W')$$

Since $P_{\Pi}$ is an MSO-property, it has a finite state
property of $t$-boundaried
graphs~\cite{Cou90}. \todo[disable,inline]{What is an extended MSO
  property, and why is this implication true? We should give a
  reference. Is~\cite{Cou90} a valid reference for this? --
  Philip}That is, there exists a finite set $\mc{S}\subseteq
\mc{H}_t$ with the property that for every pair $(G,W)\in
\mc{H}_t$, there exists a pair $(G',W')\in\mc{S}$ with
$(G,W)\sim_{\Pi} (G',W')$.

Let $G_p$ be a $t$-boundaried graph. By the definition of strong
monotonicity, there exists $W_{G_p}\subseteq V(G_p)$ such that for
every $(G,W)\in \mc{H}_t$ with finite value $\zeta_{G_p}(G,W)$,
$P_{\Pi}(G_p\oplus G,W_{G_p}\cup W)$ holds, and
$|W_{G_p}|\leqslant \zeta_{G_p}(G,W)+f(t)$. Observe also that by
definition of the function $\zeta_{G_p}$,
$\zeta_{G_p}(G,W)\leqslant |W_{G_p}|$. It follows that
\begin{equation}\label{eq:zeta}
|W_{G_p}|-f(t) \leqslant \zeta_{G_p}(G,W)\leqslant |W_{G_p}|
\end{equation}

We define the equivalence relation $\sim_{\mc{R}}$ on
$t$-boundaried graphs as follows: $G_p\sim_{\mc{R}} G_r$ if and
only if there exist sets $W_{G_p}\subseteq V(G_{p})$ and
$W_{G_r}\subseteq V(G_{r}$ meeting the condition of strong
monotonicity such that for every $(G,W)\in\mc{S}$ we have
\begin{equation} \label{eq:cr}
|W_{G_p}|-\zeta_{G_p}(G,W)=|W_{G_r}|-\zeta_{G_r}(G,W)
\end{equation}.

By~(\ref{eq:zeta}) and the finiteness of $\mc{S}$, there exists a set $\mathcal{R}$ of at most $(f(t)+2)^{|\mc{S}|}$ $t$-boundaried graphs such that for every $t$-boundaried graph $G_p$, there exists $G_r\in\mathcal{R}$ with $G_p\sim_{\mc{R}} G_r$.

Let $G_p$ and \(G_{r}\) be $t$-boundaried graphs such that
$G_p\sim_{\mc{R}} G_r$. As a consequence of (\ref{eq:cr}), there
is a constant $c_r=|W_{G_{p}}|-|W_{G_{r}}|$ (which depends only on
$G_p$ and $G_r$) such that $\zeta_{G_p}(G,W)=\zeta_{G_r}(G,W)+c_r$
for every $(G,W)\in \mc{S}$. The rest of the proof is devoted to
the following claim:

\begin{claim} \label{cl:FII} For two $t$-boundaried graphs $G_p$
  and $G_r$, if $G_p\sim_{\mc{R}} G_r$ then $G_p\equiv_{\Pi^d}
  G_r$. Specifically, for every $t$-boundaried graph $G$ and $S\in
  V(G)\setminus \partial(G)$, we have
$$(G_p\oplus G, S,k)\in \Pi^d \mbox{ if and only if~} (G_r\oplus G,S,k-c_r)\in \Pi^d$$
\end{claim}
\begin{proofof}
  We only prove the forward direction, the reverse follows with
  symmetric arguments.
  \todo[disable,inline]{This symmetry is \textbf{not}
    obvious to me, since the above relation, \emph{and} the
    relation \(\equiv_{\Pi^d}\) are not symmetric. -- Philip}
  Suppose that $(G_p\oplus G, S,k)\in \Pi^d$. Consider $Z\subseteq
  V(G_p\oplus G)$ such that $Z\cap S=\emptyset$,
  $P_{\Pi}(G_p\oplus G,Z)$ is satisfied and $Z$ has the minimum
  size. We denote $W=Z\cap V(G)$ and $W_p=Z\setminus W$. Observe
  that since \(P_{\Pi}(G_p\oplus G,Z)\) holds, $P_{\Pi}(G_p\oplus
  G,W_p\cup W)$ also holds.

  Let us consider $(G',W')\in\mc{S}$ such that $(G,W)\sim_{\Pi}
  (G',W')$. We first prove that $|W_p|=\zeta_{G_p}(G',W')$. Since
  $P_{\Pi}(G_p\oplus G,W_p\cup W)$ holds and $(G,W)\sim_{\Pi}
  (G',W')$, we have that $P_{\Pi}(G_p\oplus G',W_p\cup W')$
  holds. Hence \(|W_{p}|\geq\zeta_{G_p}(G',W')\). For the sake of
  contradiction, assume that there exists $W'_p\subseteq V(G_{p})$
  such that $|W'_p|<|W_p|$ and $P_{\Pi}(G_p\oplus G',W'_p\cup W')$
  holds. Since $(G,W)\sim_{\Pi} (G',W')$, $P_{\Pi}(G_p\oplus
  G,W'_p\cup W)$ is satisfied.  As $W\cap W_p=\emptyset$, we have
  $|W'_p\cup W|<|Z|$; this contradicts the choice of $Z$.

  Since $G_p\sim_R G_r$ and $(G',W')\in\mc{S}$, there exists
  $W_r\subseteq V(G_r)$ such that $P_{\Pi}(G_r\oplus G',W_r\cup
  W')$ holds and $|W_r|=|W_p|-c_r$. And finally, $(G,W)\sim_{\Pi}
  (G',W')$ implies that $P_{\Pi}(G_r\oplus G,W_r\cup W)$.

To conclude the proof observe first that $S\subseteq V(G)\setminus \partial(G)$ implies that $(W_r\cup W)\cap S=\emptyset$. Moreover we have
$$|W_r\cup W|\leqslant |W_r|+|W|=|W_p|-c_r+|W|= |Z|-c_r\leqslant k-c_r$$

It follows that $(G_r\oplus G,S,k-c_r)\in \Pi^d$.
\end{proofof}

By Claim~\ref{cl:FII}, we conclude that $\mc{R}$ is a set of representatives for $(\Pi^d, t)$ and thus the disjoint version $\Pi^d$ of a strongly monotone \mMSO{} problem $\Pi$ has disjoint finite integer index.
\end{proof}

\todo[disable,inline]{BLABLABLA introducing the idea of protrusion reduction rules, add a figure of a protrusion}

\begin{definition} \label{def:protrusion}
A subset $X$ of the vertex set of a graph $G$ is a \emph{$t$-protrusion} of $G$ if $tw(G[X])\leqslant t$ and $|\partial(X)|\leqslant t$.
\end{definition}


\begin{lemma}\label{safe:protrusion}
Let $\Pi^d$ be the disjoint version of a strongly monotone p-\textsc{min-MSO} problem $\Pi$. There exists a computable function $\gamma:\mathbb{N}\rightarrow \mathbb{N}$ and an algorithm that given:
\begin{itemize}
\item an instance $(G,S,k)$ of $\Pi^d$ such that $P_{\Pi}(G,S)$ holds
\item a $t$-protrusion $X$ of \(G\) such that $|X|>\gamma(2t+1)$ and
  $X\cap S=\emptyset$
\end{itemize}
in time $O(|X|)$ outputs an instance $(G',S,k')$ such that $|V(G')|<|V(G)|$, $k'\leq k$,  $(G',S,k')\in\Pi^d$ if and only if $(G,S,k)\in\Pi^d$, and $P_{\Pi}(G',S)$ holds.
\end{lemma}
\begin{proof}
  Let $\sim_{\mathcal{R}}$ be the equivalence relation on
  $(2t+1)$-boundaried graphs defined in the proof of
  Lemma~\ref{lem:disjointFII}.
  We refine
  the equivalence relation $\sim_{\mc{R}}$ into $\sim_{\mc{R}^*}$
  according to whether a $(2t+1)$-boundaried graph satisfies $P_{\Pi}$.
  be precise, we have $G_p\sim_{\mc{R}^*} G_r$ if and only if (a)
  $G_p\sim_{\mc{R}} G_r$ {\em and} (b) for every $(2t+1)$-boundaried
  graph $H$: $P_{\Pi}(G_p\oplus H)$ 
  if and only $P_{\Pi}(G_r\oplus H)$
  We know that $\sim_{\mc{R}}$
  has finite index. 
  As $P_{\Pi}$ is an MSO-expressible
  graph property, the equivalence relation (b) has finite
  index~\cite{Cou90}. Therefore $\sim_{\mc{R}^*}$ also defines
  finitely many equivalence classes. We select a set $\mc{R}^*$ of
  representatives for $\sim_{\mathcal{R}^*}$ with one further
  restriction: Claim \ref{cl:FII} is satisfied for some {\em
    nonnegative} constant $c_r$. Such a set of representatives
  $\mc{R}^*$ can be constituted by picking up a representative
  $G_r$ for each equivalence class so that the constant
  $\zeta_{G_p}(G,W)-\zeta_{G_r}(G,W)$, following the condition
  (a), is nonnegative for every $G_p\sim_{\mc{R}^*} G_r$. Here
  $\zeta$ is the signature function for $\Pi$. Define $\gamma(2t+1)$
  to be the size of the vertex set of the largest graph in
  $\mathcal{R}^*$.


  Let $\phi$ and $\rho$ be mappings from the set of
  $(2t+1)$-boundaried graphs of size at most $2\gamma(2t+1)$ to
  $\mc{R}^*$ and $\mathbb{N}$ respectively such that for every
  $(2t+1)$-boundaried graph $G$ and $S\subseteq
  V(G)\setminus \partial(G)$, we have $(G_p\oplus G,S,k)\in \Pi^d$
  if and only if $(\phi(G_p)\oplus G,S,k-\rho(G_p))\in
  \Pi^d$. Such mappings exist: we take $\phi(G_p):=G_r \in
  \mc{R}^*$ such that $G_p\sim_{\mc{R}^*} G_r$, and
  $\rho(G_p):=\zeta_{G_p}(G,W)-\zeta_{\phi(G_p)}(G,W)$ which is a
  constant by the definition of $\sim_{\mc{R}}$ (and thus of
  $\sim_{\mc{R}^*}$) and nonnegative by the way we constitute
  $\mc{R}^*$ as explained in the previous paragraph.



  Suppose that $|X|>\gamma(2t+1)$. We build a nice
  tree-decomposition of $G[X]$ of width \(t\) in $O(|X|)$ time and
  identify a bag $b$ of the tree-decomposition farthest from its
  root such that the subgraph $G_b$ induced by the vertices
  appearing in bag $b$ or below contains at least $\gamma(2t+1)$ and
  at most $2\gamma(2t+1)$ vertices. The existence of such a bag is
  guaranteed by the properties of a nice tree decomposition. Note
  that for any $X'\subset X$, we have $X'\cap S=\emptyset$. Let
  \(X'=V(G_{v})\), so that 
  that $|X'|\leq 2\gamma(2t+1)$. We replace $G[X]$ by
  $\phi(G[X'])$ \todo[disable,inline]{The function \(\phi()\) is
    defined only for \(2t\)-boundaried graphs. It is not clear why
    \(G[X']\) is a \(2t\)-boundaried graph. In particular, nothing
    seems to rule
    out the possibility that \(|\partial(X')|=2t+1\). -- Philip\\
    We have fixed this by replacing 2t with 2t+1 everywhere. --
    Philip} to obtain $G'$, and decrease $k$ by $\rho(X')$. It
  follows that $(G,S,k)\in \Pi^d$ if and only if $(G',S,k')\in
  \Pi^d$. Observe that $k'=k-\rho(X')\leq k$ and $|V(G')|<|V(G)|$
  as $|\phi(G[X])|\leq \gamma(2t+1)<|X|$. Finally, observe that the
  condition (b) of $\sim_{\mc{R}^*}$ ensures that $G'-S$ is
  $K_4$-minor-free. This completes the proof.
\end{proof}


As a corollary, since the \KD{} is strongly monotone, the following reduction rule for $\KCC{}$ is safe. We state the rule for an arbitrary value of $t$, but in practice, our reduction rule will only be based on $t$-protrusions for $t\leqslant 4$.

\setcounter{reduction}{0}
\begin{reduction} \textup{\textbf{(Generic disjoint protrusion rule)}}
  Let $(G,S,k)$ be an instance of \KCC{} and $X$ be a $t$-protrusion such that $X\cap S=\emptyset$. Then there exists a computable function $\gamma(.)$ and an algorithm which computes an equivalent instance in time $O(|X|)$ such that $G[S]$ and $G'[S]$ are isomorphic, $G'-S$ is $K_4$-minor-free, $|V(G')|<|V(G)|$ and $k'\leqslant k$, provided $|X|>\gamma(2t+1)$.
\end{reduction}

We remark on Reduction rule \ref{rrule:protrusion} that
$|\partial(X')|$ may be strictly smaller than $2t+1$. In that
case, we can identify some vertices of $X'\setminus \partial(X')$
as boundary vertices and construe $X'$ as $(2t+1)$-boundaried
graph. This is always possible for $|X'|>\gamma(2t+1)\geq
2t$. \todo[disable,inline]{Again, we do not seem to consider the
  case when
  $|\partial(X')|$ is strictly larger than \(2t\).\\
  The change from 2t to 2t+1 has fixed this problem also. --
  Philip }


\section{Deferred proof of Lemma \ref{lem:compression}}
\begin{reminder}{Lemma \ref{lem:compression}}
If \KCC{} can be solved in $c^k\cdot n^{O(1)}$ time, then \KD{} can be
solved in $(c+1)^k\cdot n^{O(1)}$ time.
\end{reminder}
\begin{proof}
  Let \(\mc{A}\) be an FPT algorithm which solves the \KCC{}
  problem in $c^k\cdot n^{O(1)}$ time. Let $(G,k)$ be the input
  graph for the \KD{} problem and let $v_1,\ldots ,v_n$ be any
  enumeration of the vertices of $G$. Let $V_i$ and $G_i$
  respectively denote the subset $\{v_1\dots v_i\}$ of vertices
  and the induced subgraph $G[V_i]$. We iterate over $i=1,\ldots
  ,n$ in the following manner. At the $i$-th iteration, suppose we
  have a \KC{} $S_i \subseteq V_i$ of \(G_{i}\) of size at most
  $k$.  At the next iteration, we set $S_{i+1}:=S_i\cup
  \{v_{i+1}\}$ (notice that \(S_{i+1}\) is a \KC{} for $G_{i+1}$
  of size at most $k+1$). If $|S_{i+1}|\leq k$, we can safely move
  on to the $i+2$-th iteration. If $|S_{i+1}|=k+1$, we look at
  every subset $S \subseteq S_{i+1}$ and check whether there is a
  \KC{} $W$ of size at most $k$ such that $W\cap
  S_{i+1}=S_{i+1}\setminus S$. To do this, we use the FPT
  algorithm $\mc{A}$ for $\KCC{}$ on the instance $(H,S)$ with
  $H=G_{i+1}-(S_{i+1}\setminus S)$. If $\mc{A}$ returns a \KC{}
  $W$ of $H$ with $|W|< |S|$, then observe that the vertex set
  $(S_{i+1}\setminus S) \cup W$ is a \KC{} of $G$ whose size is
  strictly smaller than $S_{i+1}$. If there is a \KC{} of
  $G_{i+1}$ of size strictly smaller than $S_{i+1}$, then for some
  $S\subseteq S_{i+1}$, there is a small $S$-disjoint \KC{} in
  $G_{i+1}-(S_{i+1}\setminus S)$ and $\mc{A}$ correctly returns a
  solution.

  The time required to execute $\mc{A}$ for every subset $S$ at
  the $i$-th iteration is $\sum_{i=0}^{k+1} \binom{k+1}{i} \cdot
  c^i \cdot n^{O(1)}=(c+1)^{k+1}\cdot n^{O(1)}$. Thus we have an
  algorithm for \KD{} which runs in time $(c+1)^k \cdot n^{O(1)}$.
\end{proof}

\section{Deferred proofs for (explicit) reduction rules}

\setcounter{lemma}{12}
\begin{lemma}~\label{safe:deg}
Reduction rules~\ref{rrule:deg1S},~\ref{rrule:deg2} and~\ref{rrule:para} are safe and  can be applied in polynomial time.
\end{lemma}
\begin{proof}
It is not difficult to see that each of these rules can be applied
in polynomial time. We now prove that each of them is safe.

\noindent
\emph{Reduction rule~\ref{rrule:deg1S}.} Let \(W\) be a small
\(S\)-disjoint \(K_{4}\)-minor cover of \(G\). Observe that $G'-(W\setminus X)$ is a subgraph of $G-W$. It follows
that \((W\setminus X)\) is a small \(S\)-disjoint \(K_{4}\)-minor
cover of \(G- X\). By the same reasoning, \((W\setminus
(X\setminus\partial(X)))\) is a small \(S\)-disjoint
\(K_{4}\)-minor cover of \(G-(X\setminus\partial(X))\).

For the opposite direction, let \(W'\) be a small \(S\)-disjoint
\(K_{4}\)-minor cover of \(G':=(G-X)\). Then \(G'-W'\) is
$K_4$-minor-free. Since \(G- W'\) is a disjoint union of
\(G'\setminus W'\) and \(G[X]\) and any $K_4$-subdivision is
biconnected, $G-W'$ is $K_4$-minor-free as well. Thus \(W'\) is a
small \(S\)-disjoint \(K_{4}\)-minor cover of \(G\). The same
argument goes through when \(G'=(G\setminus
(X\setminus\partial(X)))\), as well.

\noindent
\emph{Reduction rule~\ref{rrule:deg2}.} Let \(W\) be a small
\(S\)-disjoint \(K_{4}\)-minor cover of \(G\). Without loss of
generality, assume that the vertex \(v\) is \emph{not} in
\(W\). Indeed, any $K_4$-subdivision containing $v$ also contains
$u_2$ and thus, we can take $(W\setminus \{v\})\cup \{u_2\}$ to
hit such a $K_4$-subdivision. Let \(G'\) be the graph obtained
from \(G\) by applying the rule. Observe that \(G_{2}=G'\setminus
W\) is a minor of \(G_{1}=G\setminus W\), that is:
\begin{itemize}
\item If \(W\cap\{u_{1},u_{2}\}=\emptyset\), then \(G_{2}\) can be
  obtained from \(G_{1}\) by contracting the edge \((v,u_{1})\).
\item If \(W\cap\{u_{1},u_{2}\}\neq\emptyset\), then \(G_{2}\) can be
  obtained from \(G_{1}\) by deleting \(v\).
\end{itemize}
It follows that \(W\) is a
small \(S\)-disjoint \(K_{4}\)-minor cover of \(G'\) as well. For
the opposite direction, let \(W'\)
be a small \(S\)-disjoint \(K_{4}\)-minor cover of \(G'\). Observe
that \(G_{1}'=G\setminus W'\) can be obtained from the
\(K_{4}\)-minor-free graph
\(G_{2}'=G'\setminus W'\) in the following ways:
\begin{itemize}
\item If \(W'\cap\{u_{1},u_{2}\}=\{u_{1},u_{2}\}\), then \(G_{1}'\) can be
  obtained from \(G_{2}'\) by adding an isolated vertex \(v\).
\item If \(W'\cap\{u_{1},u_{2}\}=\{u_{2}\}\), then \(G_{1}'\) can be
  obtained from \(G_{2}'\) by attaching a vertex \(v\) to \(u_{1}\).
\item If \(W'\cap\{u_{1},u_{2}\}=\emptyset\), then \(G_{1}'\) can be
  obtained from \(G_{2}'\) by subdividing the edge \((u_{1},u_{2})\).
\end{itemize}
In the first two cases, note that any $K_4$-subdivision is
biconnected and thus $v$ is never contained in a
$K_4$-subdivision. By the assumption that $G'_2$ is
$K_4$-minor-free, $G'_1$ is also $K_4$-minor-free. In the third
case, $G'_1$ is also $K_4$-minor-free since subdividing an edge in
a \(K_{4}\)-minor-free graph does not introduce a \(K_{4}\)
minor. It follows that \(W'\) is a small \(S\)-disjoint
\(K_{4}\)-minor cover of \(G\) as well.

\noindent
\emph{Reduction rule~\ref{rrule:para}.} In the forward direction,
observe that the graph obtained by applying the rule is a subgraph
of the original graph. In the reverse direction, observe that
increasing the multiplicity (number of parallel edges) of any edge
in a \(K_{4}\)-minor-free graph does not introduce a
\(K_{4}\)-minor in the graph.
\end{proof}

\begin{figure}[htbh]
\begin{center}
\includegraphics[scale=0.7]{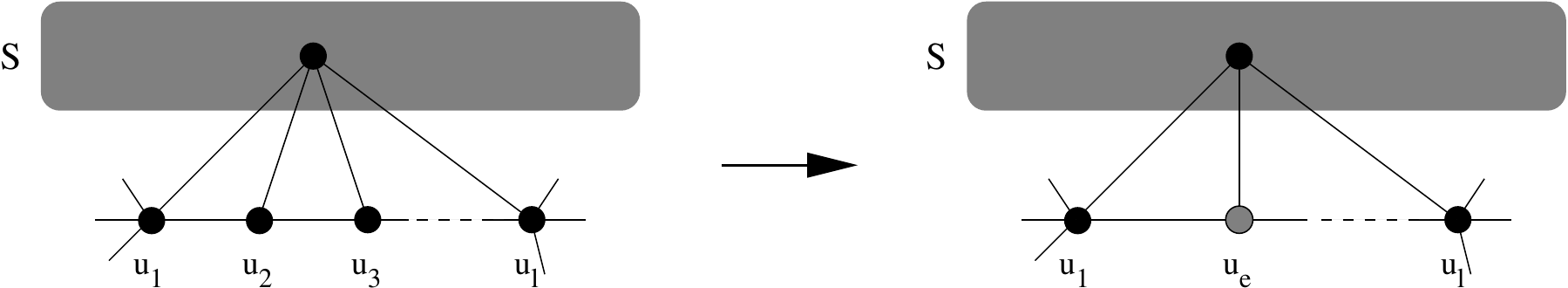}
\end{center}
\caption{Contraction of the edge $e=u_2u_3$ into $u_e$ (the grey vertex) when Reduction rule~\ref{rrule:consec} applies.}
\label{fig:rule-consec}
\end{figure}

\begin{lemma}\label{safe:consec}
Reduction Rule~\ref{rrule:consec} is safe and can be applied in polynomial time.
\end{lemma}
\begin{proof}
Let $u_e$ be the vertex obtained by contracting $e$, and let $W$ be a small disjoint \KC{} of $G$. If $W\cap\{u_{2},u_{3}\}=\emptyset$, then let $W'\gets W$; otherwise let $W'\gets(W\setminus \{u_2,u_3\})\cup\{u_e\}$. In either case $\vert W'\vert\leqslant \vert W\vert\leqslant k$, and $(G/e)\setminus W'$ is a minor of $G\setminus W$. Since $G\setminus W$ is $K_4$-minor-free, so is $(G/e)\setminus W'$, and so $W'$ is a small disjoint \KC{} of $G/e$.

Conversely, let $W'$ be a small disjoint \KC{} of
$G/e$. We first consider the case $u_e\in W'$. Then let $W\gets
(W'\setminus\{u_e\})\cup \{u_2\}$. We claim that \(W\) is a small
disjoint \(K_{4}\)-minor cover of \(G\). It is not difficult to
see that \(W\) is both small and \(S\)-disjoint; we now show that
it is a \(K_{4}\)-minor cover of \(G\). Assume to the contrary
that $G- W$ contains a $K_{4}$-subdivision $H$. Observe that
$G-(W\cup\{u_{3}\})$ is isomorphic to $(G/e)- W'$ which is
$K_4$-minor-free, and so $u_3\in V(H)$. Now $u_3$ is a degree
$2$ vertex in $G-W$ and so is a subdividing node of $H$, implying
that $u_4$ and $x$ (the neighbors of $u_3$) belongs to $V(H)$. As
$x$ and $u_4$ are adjacent, $G-W$ contains a $K_4$-subdivision
$H'$ with $V(H')=V(H)\setminus\{u_3\}$. Thus $G-(W\cup\{u_{3}\})$
contains a \(K_{4}\)-subdivision, a contradiction. 

Suppose now that $u_e\notin W'$. We claim that $W'$ is a
\KC{} of $G$ as well. Assume to the contrary that $H$ is a
$K_4$-subdivision in $G- W'$. We claim that every
$K_4$-subdivision $H$ in $G-W'$ contains $u_2$ and $u_3$ as
branching nodes. Assume that $u_2\notin V(H)$. Then since
$G-(W'\cup\{u_2\})$ is a (non-induced) subgraph of $G/e-W'$, $H$
is also a $K_4$-subdivision in $G/e-W'$: a contradiction. So every
$K_4$-subdivision in $G-W'$ contains $u_2$. By a symmetric
argument, \(u_{3}\in V(H)\) as well. %
\todo[disable,color=lightgray,inline]{I have replaced
  the earlier argument by an appeal to symmetry. I cannot now see
  why this is wrong, but since we had a different, more
  complicated argument before, it is possible that this new
  argument is wrong. This is the previous argument : ``If
  $u_3\notin V(H)$, then $u_2$ has to be a subdividing node of $H$
  with neighbors $x$ and $x_1$. Observe that since $u_1$ and $x$
  are adjacent, $G-W'$ also contains a $K_4$-subdivision $H'$
  avoiding $u_2$: contradiction, so $u_2,u_3\in V(H)$.'' ---
  Philip}
\todo[disable,inline]{Christophe says: ``I had some doubt about validity
  of the use of symmetry (Eunjung also used symmetry). My concern
  was the folllowing: we need to argue that both $u_2$ and $u_3$
  belong to $V(H)$. I feel (but maybe wrong) that the symmetry
  only guarantee that $u_2$ or $u_3$ belong to $V(H)$. So I argued
  that assuming $u_2$ belong to $V(H)$, $u_3$ also has to be
  in. What do you think ?''

  Since our argument which says \(u_{2}\in V(H)\) holds
  irrespective of whether \(u_{3}\in V(H)\), I feel the symmetry
  argument is valid. We can always put the other one back if we
  feel otherwise. --- Philip}
 Now a simple case by case analysis (see
Figure~\ref{fig:case-consec}) shows that if $u_2$ or $u_3$ is a
subdividing node, then $G/e-W'$ also contains a $K_4$-subdivision
$H'$ with $V(H')=(V(H)\setminus\{u_2,u_3\})\cup\{u_e\}$: a
contradiction.

\begin{figure}[htbh]
\begin{center}
\includegraphics[scale=0.7]{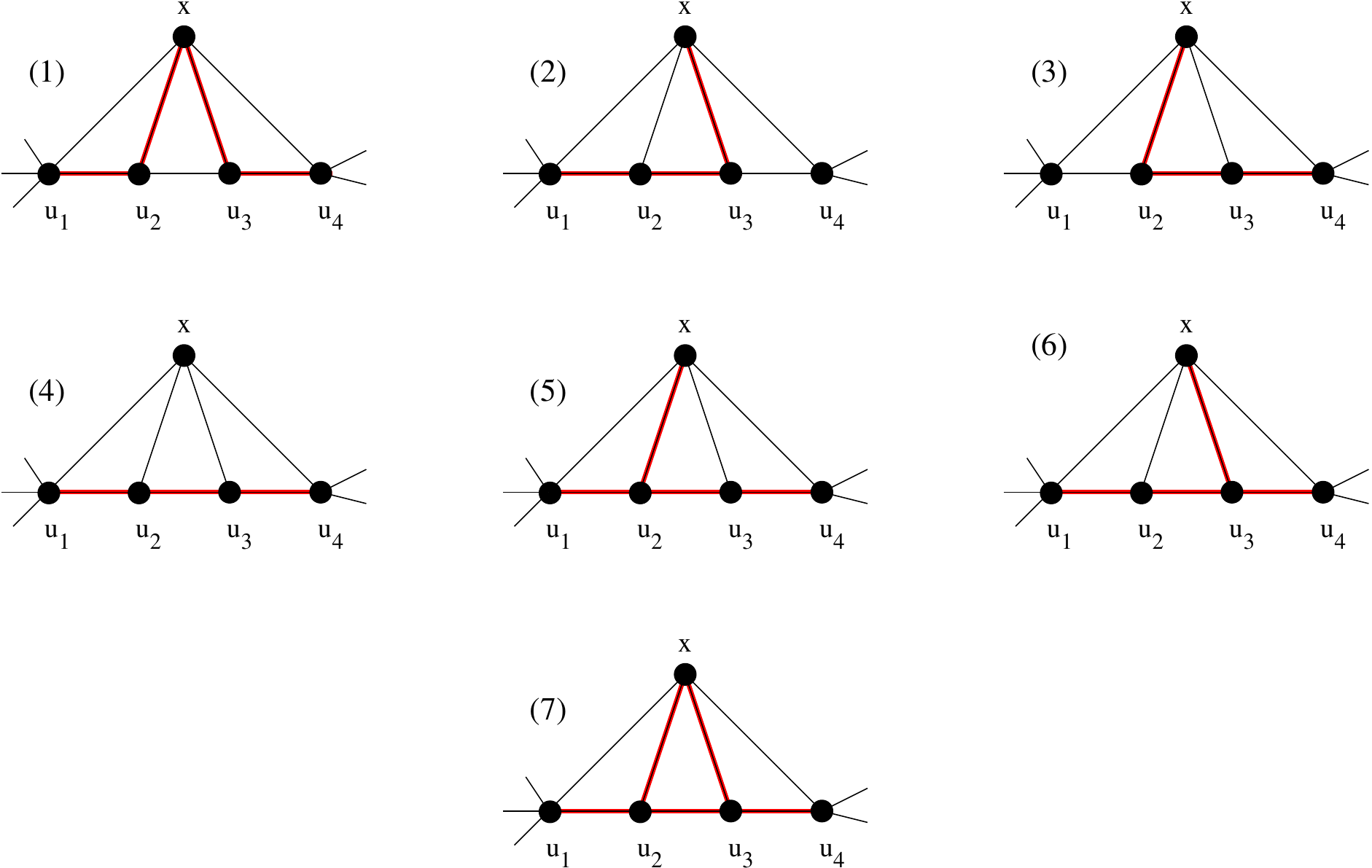}
\end{center}
\caption{The different possible intersections of $H$ with
  $G[\{u_1,u_2,u_3,u_4,x\}]$. The bold lines denote those edges in
  \(H\) which are incident on \(u_{2}\) or \(u_{3}\). In cases
  (1), (2) and (3) we can argue that there exists a
  $K_4$-subdivision in $G-W'$ avoiding either $u_2$ or $u_3$: a
  contradiction. In cases (4), (5) and (6), we observe the
  existence of a $K_4$-subdivision in $G/e$: a contradiction.}
\label{fig:case-consec}
\end{figure}

\todo[disable,color=lightgray,inline]{Do you think we need to detail the arguments for the
  above mentioned case by case analysis ?

  I think these pictures
  are good enough, and that no more detail is required. --- Philip}

It follows that $u_{2},u_{3}$ are both present as branching nodes
in $H$ (see case (7) in Figure~\ref{fig:case-consec}). As these
vertices both have degree $3$ in $G$, every edge incident to $u_2$
or $u_3$ is used in $H$. Therefore the common neighbor $x$ of
$u_2$ and $u_3$ also appears in $H$ as a branching node. So at
most one vertex in $\{u_{1},u_{4}\}$ is a branching node; assume
without loss of generality that $u_{4}$ is a subdividing node. It
lies on the path between $u_3$ and a branching node
$y\notin\{u_2,u_3,x\}$, and we can make $u_{4}$ a branching node
instead of $u_{3}$ to obtain a new $K_{4}$-subdivision $H'$ by
replacing in $H$ the edge $(x,u_3)$ by the edge $(x,u_4)$. But
then $H'$ is a $K_4$-subdivision in $G\setminus W'$ which does not
contain $u_{3}$ as a branching node, a contradiction. It follows
that $W'$ is a small disjoint \KC{} of $G$.

It is not difficult to see that the rule can be applied in polynomial time.
\end{proof}

\begin{lemma}\label{safe:2cut}
Let $(G,S,k)$ be an instance reduced with respect to Reduction Rules~\ref{rrule:deg1S},~\ref{rrule:deg2} and~\ref{rrule:para}. Then Reduction Rule \ref{rrule:2cut} is safe and can be applied in polynomial time.
\end{lemma}
\begin{proof}
  Since $(G,S,k)$ is reduced with respect to
  Rule~\ref{rrule:deg1S}, $G[F]$ does not contain any cut vertex.
  Let $(G',S,k)$ be the instance obtained by applying Reduction
  Rule \ref{rrule:2cut} to $(G,S,k)$. Let $X'$ be the set of
  vertices with which the rule replaced $X$ and let
  $X_0:=X\setminus \{s,t\}$, $X_0':=X'\setminus \{s,t\}$. We can
  assume that $X_0\neq\emptyset$ since otherwise the reduction
  rule is useless. To prove that $(G,S,k)$ has a small disjoint
  \KC{} of $G$ if and only if $(G',S,k)$ does, we need
  the following claim.

\begin{numclaim}\label{2cut:1}
$G[X]+(s,t)$ is an SP-graph if and only if $G[X]+(s,t)$ is $K_4$-minor-free.
\end{numclaim}
\begin{proofof}
  The forward direction follows directly from
  Lemma~\ref{SPcharac}. Assume now that $G[X]+(s,t)$ is
  $K_4$-minor-free. As $(G,S,k)$ is reduced with respect to
  Reduction Rule~\ref{rrule:deg1S}, the block tree of $G[X]$ is a
  path and moreover $s$ and $t$ belong to the two leaf blocks,
  respectively (these blocks may also coincide). This implies that
  the addition of the edge $(s,t)$ to $G[X]$ yields a biconnected
  graph. This concludes the proof since by Lemma \ref{SPcharac} a biconnected $K_4$-minor
  free graph is an SP-graph.
\end{proofof}

\todo[disable,color=lightgray,inline]{What happens in the forward
  direction is the following: In both the cases, the new graph
  \(G'\) is a minor of the graph \(G\), and so any \(K_{4}\)-minor
  cover of \(G\) is also a \(K_{4}\)-minor cover of \(G'\). The
  fact that \(G'\) is a minor of \(G\) is easy to show in the
  first case, but for the second case I cannot clinch the
  argument. What we need in the second case is the following
  lemma: ``Let \(G\) be a \(K_{4}\)-minor-free graph, and let
  \(s,t\) be two vertices in \(G\) such that adding the edge
  \((s,t)\) to \(G\) yields a graph which contains a
  \(K_{4}\)-minor. Then \(G\) contains the bottom right graph as a
  minor, where \(s\) and \(t\) are mapped to the two vertices of
  degree three.'' --- Philip}

\medskip We now resume the proof of the lemma. Let $W$ be a small
disjoint \KC{} of $G$. If $W\cap X_0 \neq \emptyset$,
set $W^*:=(W\setminus X_0)\cup \{t\}$. Since $\{s\}$ is a cut
vertex in $G-W^*$ isolating $X_0$, no $K_4$-subdivision in $G-W^*$
uses any vertex from $X_0$. Also $|W^*|\leq k$, and so
$W^{*}\subseteq V(F)\setminus X_{0}$ is a small disjoint
\KC{} of $G$. So we can assume without loss of generality
that $W\cap X_0= \emptyset$. Let us prove that $W$ is a
\KC{} of $G'$. For the sake of contradiction, let $H'$
be a $K_4$-subdivision in $G'-W$. There are two cases to consider:

\begin{enumerate}
\item Reduction Rule~\ref{rrule:2cut} replaces $G[X]$ by the edge
  $(s,t)$: Observe that all the branching nodes of $H'$ belong to
  $V(G)\setminus (W\cup X_0)$. Suppose $H'$ uses the edge $(s,t)$
  for a path between two branching nodes, say $u$ and $v$. As
  $W\cap X_0= \emptyset$, using an arbitrary $s,t$-path $P$ in
  $G[X]$ instead of the edge $(s,t)$ witnesses the existence of a
  $u,v$-path $G-W$. This implies that $G-W$ contains a
  $K_4$-subdivision $H$ such that $V(H)=V(H')\cup V(P)$, a
  contradiction.

\item Reduction Rule~\ref{rrule:2cut} replaces $G[X]$ by a
  $\theta_3$ on vertex set $X'=\{a,b,s,t\}$: this occurs when
  $G[X]+(s,t)$ is not an SP-graph and so by Claim~\ref{2cut:1}
  contains a $K_4$-subdivision. By Observation~\ref{cross2cut},
  the branching nodes of $V(H')$ belong either to $X'$ or to
  $V(G)\setminus \{a,b\}$. In the latter case, vertex $a$ or $b$
  may be used by $H'$ as a subdividing node to create a path
  through $s$ and $t$ between two branching nodes of $H'$. The
  same argument as above then yields a contradiction.  In the
  former case, observe that every vertex of $X'$ is a branching
  node of $H'$ and some vertices out of $X'$ may be used by $H'$
  as subdividing nodes to create the missing path $P$ between $s$
  and $t$ in $G'-W$. As $G[X]+(s,t)$ also contains a
  $K_4$-subdivision, say $H$, we can construct a $K_4$-subdivision
  in $G-W$ on vertex set $V(H)\cup V(P)$, a contradiction.
\end{enumerate}

For the reverse direction, let $W'$ be a small disjoint
\KC{} of $G'$. Again we can assume that $W'\cap
X_0'=\emptyset$. Indeed, if $W'\cap X_0'\neq\emptyset$, it is easy
to see that $(W'\setminus X_0')\cup\{t\}$ is also a small disjoint
\KC{} of $G'$.
Let us prove that $W'$ is also a \KC{} of $G$ (the
arguments are basically the same as above).  For the sake of
contradiction, assume $H$ is a $K_4$-subdivision of $G-W'$. By
Observation~\ref{cross2cut}, since $\{s,t\}$ is a separator of size two,
the branching nodes of $V(H)$ belong either to $X$ or to
$V(G)\setminus X_0$. In the former case, $G[X]+(s,t)$ is not
an SP-graph, and thus $X$ as been replaced by a $\theta_3$ on
$\{a,b,s,t\}$. Let $P$ be the $s,t$-path of $G-(X_0\cup W')$ used
by $H$. As $W'\cap X_0'=\emptyset$, $\{a,b,s,t\}\cup V(P)$ induces
a $K_4$-subdivision in $G'-W'$, a contradiction. In the latter
case, if $H$ uses a path between $s$ and $t$ in $G[X]-W'$,
then such a path also exists in $G'-W'$ witnessing a
$K_4$-subdivision in $G'-W'$, a contradiction.
\end{proof}

\section{Deferred proofs of Lemmas \ref{lem:K4-3+1} and \ref{cor:simplified-deg2}}

\begin{reminder}{Lemma \ref{lem:K4-3+1}} Let $W$ and $Z$ be disjoint
  vertex subsets of a graph $G$ such that $G[W]$
  is biconnected, $G[Z]$ is connected and $|N_W(Z)|\geq 3$. Then
  $G[W\cup Z]$ contains a $K_4$-subdivision.
\end{reminder}
\begin{proof}
  Let $x$, $y$ and $z$ be three vertices of $N_W(Z)$. Since
  \(G[Z]\) is connected and since contracting edges does not
  introduce a new \(K_{4}\)-subdivision, we may assume without
  loss of generality that there is a single vertex, say \(u\), in
  \(Z\) such that \(\{x,y,z\}\subseteq N(u)\).

  Since \(G[W]\) is biconnected, it follows from Menger's Theorem
  that there are at least two distinct paths in \(G[W]\) between
  any two vertices in \(W\). Therefore, every pair of vertices in
  \(W\) belong to at least one cycle of \(G[W]\).

  Let \(C\) be a cycle in \(G[W]\) to which \(x\) and \(y\)
  belong. If \(z\) also belongs to \(C\), then the subgraph \(G[C
  \cup \{u\}]\) contains a \(K_{4}\)-subdivision with \(x, y, z,
  u\) as the branching nodes, and we are done. So let \(z\) not
  belong to the cycle \(C\).

  Since \(G[W]\) is biconnected, \(|N_{W}(z)|\geq 2\). From
  Menger's Theorem applied to \(C\) and \(N_{W}(z)\), we get that
  there are at least two paths from \(z\) to \(C\) which intersect
  only at \(z\).  These paths together with the cycle \(C\)
  constitute a \(\theta_{3}\)-subdivision in which \(x\) and \(y\)
  are branching nodes and \(z\) is a subdividing node. Together
  with the vertex \(u\), this \(\theta_{3}\)-subdivision forms a
  \(K_{4}\) in  $G[W\cup Z]$.
  %
\end{proof}

\begin{reminder}{Lemma \ref{cor:simplified-deg2}}
If $(G,S,k)$ is a simplified instance of \KCC{}, then $F=N_0\cup N_1\cup N_2$.
\end{reminder}
\begin{proof}
As $(G,S,k)$ is a simplified instance, $G[S\cup\{x\}]$ is $K_4$-minor-free for every $x\in F$ (by Branching rule~\ref{rule:br1}) and there exists a biconnected component $B$ of $G[S]$ containing $N_S(x)$ (otherwise we could apply Branching rule~\ref{rule:br2} or~\ref{rule:br3}). It directly follows from Lemma~\ref{lem:K4-3+1}, that for every vertex $x\in F$, $|N_S(x)|\leqslant 2$.
\end{proof}

\section{Deferred proofs of Theorem \ref{th:independent} and Theorem \ref{lem:transVC}}
\begin{reminder}{Theorem \ref{th:independent}} Let $(G,S,k)$ be an
  instance of \KCC{}. If none of the reduction rules nor branching
  rules applies, then $(G,S,k)$ is an independent instance.
\end{reminder}
\begin{proof}
  Once we show that $F$ is an independent set, condition (b)
  follows from Corollary \ref{cor:simplified-deg2} and the fact
  that $(G,S,k)$ is reduced with respect to Reduction rule
  \ref{rrule:deg1S}. Conditions (c) and (d) are also satisfied in
  this case since $(G,S,k)$ is simplified, specifically since
  Branching rules~\ref{rule:br1}, \ref{rule:br2}
  and~\ref{rule:br3} do not apply on singleton sets \(X\). We now
  prove that $F$ is an independent set.

  Suppose $G[F]$ contains a connected component $X$ with at least
  two vertices. Since \((G,S,k)\) is a simplified instance,
  \(G[X\cup S]\) does not contain \(K_{4}\) as a minor. Hence from
  Lemma~\ref{lem:K4-3+1}, we have $|N_S(X)|\leq 2$. We consider
  two cases, whether $G[X]$ is a tree or not.

  Let us assume that $X$ is a tree. Observe that every leaf of $X$
  belongs to $N_2$, for otherwise Rule~\ref{rrule:deg1S} or
  Rule~\ref{rrule:deg2} would apply. So $X$ contains two leaves,
  say $u$ and $v$, having the same two neighbors in $S$, say $x$
  and $y$. But then observe that $x$ and $y$ belong to the same
  connected component of $G[S]$ (otherwise Branching
  Rule~\ref{rule:br2} would apply). It clearly follows that $x$,
  $y$, $u$ and $v$ are the four branching nodes of a
  $K_4$-subdivision in $G[S\cup X]$, which contradicts the
  assumption that Branching Rule~\ref{rule:br1} cannot apply to
  $(G,S,k)$.

  We now consider the case where $X$ is not a tree. Before we
  proceed further we observe the following. A \emph{nontrivial}
  block is a block which is more than just an edge.

\begin{claim}\label{cl:connection}
  Let $B$ be a nontrivial block of $G[F]$. Let \(F_{B}\) be
  the graph obtained from \(G[F]\) by removing
  \(B\setminus \partial_G(B)\) and all the edges in
  \(G[\partial_G(B)]\). Then every connected
  component of 
  \(F_{B}\) contains a vertex of $N_1\cup N_2$.
\end{claim}
\begin{proofof}
  Observe that any connected component of $F_{B}$ shares at most
  one vertex with $B$. Thus if a connected component of
  $G[F\setminus (B\setminus \partial_G(B))]$ is entirely contained
  in $N_0$, then we can apply Reduction rule \ref{rrule:deg1S}.
\end{proofof}

As $X$ is not a tree, it contains a non-trivial block $B$
.  Since $(G,S,k)$ is
reduced with respect to Reduction Rule~\ref{rrule:deg1S},
$|\partial_G(B)|\geqslant
2$. 

We first assume that $|\partial_G(B)|=2$ with
$\partial(B)=\{s,t\}$. Observe that $G[B]+(s,t)$ is not a
series-parallel graph since otherwise $B$ would be a single edge
$(s,t)$ due to Reduction rule~\ref{rrule:2cut}. As $(G,S,k)$ is
reduced with respect to Reduction rule~\ref{rrule:2cut}, $B$ is a
$\theta_3$ with $s$ and $t$ as subdividing nodes. Due to Branching
rule \ref{rule:br2}, $N_S(X)$ is contained in a single connected
component of $S$. Together with the observation of Claim
\ref{cl:connection}, this implies that there exists an $s,t$-path
$P$ in $G[S\cup X]$ in which no internal vertex lies in
$B$. However, $G[B\cup P]$ is a $K_4$-subdivision and Branching
rule~\ref{rule:br1} would apply, a contradiction.

So we have that $|\partial_G(B)|\geq 3$ and let
$\{x,y,z\}\subseteq \partial(B)$.  By Claim~\ref{cl:connection},
there exist three internally vertex-disjoint paths $P_x$, $P_y$
and $P_z$ from $x$, $y$ and $z$ respectively to a connected
component $G[S]$ such that no internal vertex of them lies in
$B$. Since $B$ is biconnected, Lemma~\ref{lem:K4-3+1} applies by
taking $B$ and $(S\cup P_x\cup P_y\cup P_z)\setminus \{x,y,z\}$
showing that $G[B\cup P_x \cup P_y \cup P_z \cup S]$ contains a
$K_4$-subdivision: a contradiction of the fact that Branching
rule~\ref{rule:br1} does not apply.
\end{proof}

\begin{reminder}{Theorem \ref{lem:transVC}}
Let $(G,S,k)$ be an independent instance of \KCC{}. Then $W\subseteq F$ is a disjoint \KC{} of $G$ if and only if it is a vertex cover of $G^*(S)$.
\end{reminder}
\begin{proof}

If $W\subseteq F$ is a \KC{} of $G$, then by construction $G^*(S)-W$ is an independent set and thus, $W$ is a vertex cover of $G^*(S)$. 

To show the converse, we can assume that $G[S]$ is biconnected. Indeed, for every $v\in F$, its two neighbors $x_v,y_v\in S$ belong to the same biconnected component and thus any cut vertex of $G[S]$ remains a cut vertex of $G-W$. Since $K_4$-subdivision is biconnected, any such subdivision in $G-W$ must not contain $u,v\in F\setminus W$ such that $N_S(u)$ and $N_S(v)$ belong to distinct biconnected components of $G[S]$.

An SP-tree is {\em minimal} if any S-node (resp. P-node) does not have S-nodes (resp. P-nodes) as a child \cite{BodlaenderF96}. Furthermore, any SP-tree obtained will be converted into a minimal one via standard operations on the given SP-tree: if there is an S-node (resp. P-node) with another S-node (resp. P-node) as a child, contract along the edge and if an S-node or P-node has exactly one child, delete it and connect its child and its parent by an edge. Throughout the proof, we fix a minimal SP-tree $\mathcal{T}_S$ of $G[S]$. Furthermore, we take the root as follows: (a) $G[S]$ is a cycle, we let two adjacent vertices be the terminals of the root. (2) otherwise, the last parallel operation has at least three children. 

For a node $\alpha$ of
the SP-tree $\mc{T}_S$, let $Z_\alpha$ be the set of terminals of
its children $\alpha_1\dots \alpha_c$, that is,
$Z_{\alpha}=\bigcup_{1\leqslant i\leqslant c} X_{\alpha_i}$. 

\begin{claim}\label{cl:uniqueS}
For every $u\in F$, either $X_{\alpha}=\{x_u,y_u\}$ for some node $\alpha$ of $\mc{T}_S$ or there is a unique $S$-node $\alpha$ such that $\{x_u,y_u\}\subseteq Z_{\alpha}$.
\end{claim}
\begin{proofof}
Let us suppose that for $u\in F$, there no $\alpha$ in $\mc{T}_S$ such that $X_{\alpha}=\{x_u,y_u\}$. We argue that for such $u$, there exists an $S$-node $\alpha$ such that $\{x_u,y_u\}\subseteq Z_{\alpha}$.

To this end, take a lowest node $\alpha$ such that $x_u,u_y\in V_{\alpha}$ and let $X_{\alpha}=\{s,t\}$. Then $\alpha$ should be an S-node. Suppose $\alpha$ is a P-node. As we choose $\alpha$ to be lowest, there are two children $\beta_x$ and $\beta_y$ of $\alpha$ such that $x_u\in Y_{\beta_x}$ and $y_u\in Y_{\beta_y}$. This implies $G[S]$ is not a cycle as we fix the terminals of the root to be adjacent vertices in this case. Note that $X_{\alpha}=X_{\beta_x}=X_{\beta_y}$ and $X_\alpha$ separates $x_u$ and $y_u$.

Since $G[V_{\beta_x}]$ is an SP-graph, there is a path $P_x$ from $s$ to $t$ visiting $x_u$. Likewise, $G[V_{\beta_y}]$ contains a path $P_y$ from $s$ to $t$ visiting $y_u$. On the other hand, since $G[S]$ is not a simple cycle, there is a P-node $\alpha'$ such that either (a) $\alpha'=\alpha$ and $\alpha'$ has a child $\beta\neq \{\beta_x,\beta_y\}$, or (b) $\alpha'$ is an ancestor of $\alpha$ and it has a child $\beta$ which is not an ancestor of $\alpha$. In both cases, the subgraph $G[S\setminus (Y_{\beta_x}\cup Y_{\beta_y})]$ is connected and contains a path $P$ connecting $s$ and $t$. The three paths $P_x$, $P_y$, $P$ and the length-two path between $x_u$ and $y_u$ via $u$ form a $K_4$-subdivision with $\{v_x,v_y,s,t\}$ branching nodes.


Now we argue the uniqueness of such an S-node. For some $u\in F$,
suppose that there are two distinct S-nodes $\alpha$ and $\alpha'$
such that $\{x_u,y_u\}\subseteq Z_{\alpha}$ and
$\{x_u,y_u\}\subseteq Z_{\alpha'}$. Since $X_{\alpha}$ is a
separator of $G[S]$, the only possibility is to have
$X_{\alpha}=X_{\alpha'}=\{x_u,y_u\}$. This contradicts to our
assumption that there is no vertex $u$ such that $\{x_u,y_u\}$
labels a node of $\mc{T}_S$. 
\end{proofof}

Let $F_0$ and $F_1$ form a partition of $F$: $u\in F_0$ if $X_{\alpha}=\{x_u,y_u\}$ for some node $\alpha$ of $\mc{T}_S$, otherwise $u$ belongs to $F_1$. For $u\in F_1$, we denote as $\alpha(u)$ the unique S-node of $\mc{T}_S$
with $\{x_u,y_u\}\subseteq Z_{\alpha}$.

Suppose $W\subseteq F$ is a vertex cover of $G^*(S)$. We shall then incrementally extend $\mathcal{T}_S$ to an SP-tree of $G[S]+(F\setminus W)$. For $u\in F$, let $\mathcal{T}_u$ be the minimal SP-tree with $\{x_u,y_u\}$ as terminals of the length-two path $x_uuy_u$. It is not difficult to increment $\mathcal{T}_S$ to an SP-tree $\mathcal{T}_{S+F_0}$ of $G[S\cup F_0]$. Let $u\in F_0$ and $\alpha$ be the node labeled by $\{x_u,y_u\}$. If $\alpha$ is an S-node, there is a P-node labeled by the same terminals. Hence we assume that $\alpha$ is either a leaf node or a P-node. We do the following: (1) if $\alpha$ is a P-node, make $\mathcal{T}_u$ to be a child of $\alpha$, (2) if $\alpha$ is an edge node, convert $\alpha$ into a P-node and make $\mathcal{T}_u$ to be a child of $\alpha$. The resulting SP-tree is again minimal, via standard manipulation if necessary. It is worth noting that none of S-nodes are affected during the entire manipulation and thus $\alpha(u)$ remains unaffected for $u\in F_1$.

We wish to show that $\mathcal{T}_{S+F_0}$ can be extended to contain all $F_1\setminus W$ as well. When $\alpha$ is an S-node, $Z_{\alpha}$ can be construed as an interval on the terminals of its children: the the ordering of series compositions imposes an ordering on the elements of $Z_{\alpha}$. The crucial observation is that if $\alpha(u)=\alpha(v)$ for $u,v\in F_1\setminus W$, then the intervals $[x_u,y_u]$ and $[x_v,y_v]$ in $\alpha(u)$ do not overlap. Suppose they overlap. We can take a cycle $C$ containing all the vertices of $Z_{\alpha}$. 
Then $C$ together with the two paths $P_u=x_uuy_u$ and $P_v=x_vvy_v$ form a $K_4$-subdivision in $G[C\cup\{u,v\}]$. Therefore, we have an edge $(u,v)$ in $G^*(S)$, a contradiction.

Starting from $\mathcal{T}_{S+F_0}$, now we increment the SP-tree by attaching $\mathcal{T}_u$ for every $u\in F_1\setminus W$. Given $u\in F_1\setminus W$, add a P-node $\alpha'$ with $X_{\alpha'}=\{x_u,y_u\}$ as a child of $\alpha(u)$ and make $\alpha'$ to become the father of every former child $\alpha_i$ of $\alpha$ for which $X_{\alpha_i}$ is contained in the interval $[x_u,y_u]$. Note that no S-node other than $\alpha(u)$ is affected by this manipulation. Moreover, $\alpha(u)$ remains as an S-node. Indeed, if we need to change $\alpha(u)$, it is only because $\alpha(u)$ has a unique child after the operation. This implies $x_u,y_u$ are in fact the terminals of $X_{\alpha(u)}$. However, the parent of $\alpha(u)$, which is a P-node due to minimality of the SP-tree, is labeled by $\{x_u,y_u\}$, a contradiction. Finally due to the crucial observation from the previous paragraph, this incremental extension can be performed for all vertices of $F_1\setminus W$. Implying $G-W$ is an SP-graph, this complete the proof.
\end{proof}

\todo[disable, inline]{The proof above has to be completed: the missing references refer to the circle graph lemma, which we decided to cancel in this version}

\section{Deferred proof of Lemma \ref{lem:marked-cst-size}} 



 \begin{lemma} \label{lem:withattach} Let $(G,S,k)$ be a reduced
  instance. If $\alpha$ is a non-leaf node of an extended
  SP-decomposition $(T,\mathcal{X})$ of $G[F]$
  , then $(V_{\alpha}\setminus
  Y_{\alpha})\setminus N_0 \neq\emptyset$.
\end{lemma}
\begin{proof}
  Observe that for every non-leaf node $\alpha$ of
  $(T,\mathcal{X})$, the set $Y_{\alpha}=V_{\alpha}\setminus
  X_{\alpha}$ is nonempty. This can be easily verified when
  $\alpha$ is a cut node, an edge node which is not a leaf (this
  happens only when the edge node is the parent of a cut node in
  the extended decomposition), or an S-node. When $\alpha$ is a
  P-node, the fact that $(G,S,k)$ is reduced with respect to
  Reduction Rule \ref{rrule:para} ensures $Y_{\alpha}\neq
  \emptyset$.

  For the sake of contradiction, suppose that $Y_{\alpha}\subseteq
  N_0$. Observe that no vertex in
  \(Y_{\alpha}\) has a neighbor in \(F\setminus V_{\alpha}\). By
  assumption, no vertex in \(Y_{\alpha}\) has a neighbor in
  \(S\). Hence $\partial(V_{\alpha})\subseteq X_{\alpha}$ and thus
  $Y_{\alpha}\subseteq V_{\alpha}\setminus \partial(V_{\alpha})$.
  If $|\partial(V_{\alpha})|=1$ then Reduction
  Rule~\ref{rrule:deg1S} applies, a contradiction. Thus
  $|\partial(V_{\alpha})|=2$, and so
  \(\partial(V_{\alpha})=X_{\alpha}\). Furthermore, no descendant
  of $\alpha$ is a cut node in \(G[F]\)(otherwise Reduction
  Rule~\ref{rrule:deg1S} applies), which implies that $V_{\alpha}$
  is contained in a leaf block of $G[F]$. $G_{\alpha}$ is thus a
  series-parallel graph having $X_{\alpha}=\{s,t\}$ as terminals
  and thus by Lemma~\ref{biconnectedSP} $G_{\alpha}+(s,t)$ is an
  SP-graph. Since $\alpha$ is a non-leaf node and $(G,S,k)$ is
  reduced with respect to Reduction rule \ref{rrule:para}, we have
  $|V_{\alpha}|>2$. Thus \(G_{\alpha}\) is not isomorphic to any
  of the two excluded graphs of Reduction
  Rule~\ref{rrule:2cut}. So Reduction Rule~\ref{rrule:2cut}
  applies deleting the nonempty set $Y_{\alpha}$, a contradiction.
\end{proof}

\begin{lemma}\label{lem:cst-size-block}
  Let $(G,S,k)$ be a simplified instance of \KCC{} and $\alpha$ be
  a marked node of the extended SP-decomposition $(T,
  \mathcal{X})$ of $G[F]$. Then every block $B$ in $G_{\alpha}$
  satisfies $|B|< \gamma(9)$.
\end{lemma}
\begin{proof}
  \todo[disable,inline]{Since the graph
    \(G_{\alpha}\) is not necessarily an \emph{induced} subgraph
    of \(G[F]\), it seems possible that the block structure of
    \(G_{\alpha}\) is not necessarily the same as that of
    \(G[V_{\alpha}]\). It is therefore not clear how the previous
    assertion holds. -- Philip\\

    From the discussion with EJK: This is clearly true for all
    blocks of \(G_{\alpha}\) which are also blocks in \(G\). The
    block structure of \(G_{\alpha}\) is different from that of
    \(G[V_{\alpha}]\) only on the subgraph defined by the vertices
    in the block of \(G[F]\) to which \(\alpha\) belongs. If
    \(\alpha\) is a parallel node, then Lemma~\ref{lem:2attach}
    applies. Else if it is a serial node, we can apply the proof
    of Lemma~\ref{lem:2attach} to each child of \(\alpha\).  --
    Philip}
  Recall that the root of the SP-tree of $B$ is a P-node $\beta$
  inherited from $(T,\mathcal{X})$. As a descendent of $\alpha$, $\beta$ is a marked
  node. By Lemma~\ref{lem:2attach}, $V^B_{\beta}$ is a
  $4$-protrusion. As \(\beta\) is marked, $V^B_{\beta}$ is reduced
  under protrusion rule (Reduction Rule~\ref{rrule:protrusion})
  and so $|B|\leqslant |V^B_{\beta}|< \gamma(9)$.
\end{proof}


\begin{lemma} \label{lem:simplepath} Let $(G,S,k)$ be a simplified
  instance of \KCC{} and let $\alpha$ be a marked cut node of the
  extended SP-decomposition $(T,\mathcal{X})$ of $G[F]$ with
  $X_{\alpha}=\{c\}$. Then $|V_{\alpha}|\leqslant c_0=\gamma(9)+7$. Moreover, the block tree of $G_{\alpha}$ is a
  path.
\end{lemma}

\todo[disable,inline]{I had some formal problems with this proof, as we
  are dealing here with $\vec{\mc{B}}_{G_{\alpha}}$ which is a
  totally different object than $(T, \mathcal{X})$. So the Lemmas we have so far did not strictly apply: ${\mc{X},T}$ is obtained from $\vec{\mc{B}}_G$. Isn't this enough?}

\begin{proof}

  Let $\vec{\mc{B}}_{F_{\alpha}}$ be the oriented block tree of
  $G_{\alpha}$ rooted at $B_c$, the block containing $c$. Let $B_1$ be a leaf block in $\vec{\mc{B}}_{F_{\alpha}}$ and $c_1$ be the cut vertex such that
  $(c_1,B_1)\in E(\vec{\mc{B}}_F)$. Observe that
  $(T,\mathcal{X})$ contains a cut node $\beta_1$ such that $X_{\beta_1}=\{c_1\}$ and by the construction of $(T,\mathcal{X})$, the node $\beta_1$ is a descendant of $\alpha$. By Lemma~\ref{lem:withattach}, $B_1$ contains a vertex of
  $N_1\cup N_2$, say $x_1\in B_1$ such that $x_1\neq c_1$.  We consider two cases.

  \noindent (a) $B_1$ is a nontrivial block.\\
  Consider the remaining part of $G_{\alpha}$, i.e. $C_1:=(V_{\alpha}\setminus B_1) \cup \{c_1\}$. We shall show that $C_1\subseteq N_0$, i.e. no vertex of $C_1$ has a neighbor in $S$. Suppose the contrary and observe that $G[C_1\cup S]$ contains a path $P_1$ between $c_1$ and $S$ avoiding $B_1$. If there is a vertex $y_1 \in B_1$ s.t. $y_1\notin \{ c_1,x_1\}$ and $y_1 \subseteq N_1\cup N_2$, then by Lemma \ref{lem:K4-3+1}, $G[V_{\alpha}\cup S]$ contains a
  $K_4$-subdivision, a contradiction. If no such vertex $y_1$ exists, observe that $\{x_1,c_1\}$ forms a boundary of $B_1$. Due to the assumption that $\alpha$ is marked, the subgraph $G[V_{\alpha}\cup S]$ is $K_4$-minor-free. In particular, the subgraph $G[B_1\cup P]$ is $K_4$-minor-free, where $P$ is a path between $x_1$ and $c_1$ in $G[V_{\alpha}\cup S]$ avoiding $B_1$. The existence of such $P$ is ensured due to the existence of $P_1$, that $x_1\in N_1\cup N_2$ and the fact that $N_S(V_{\alpha})$ belong to the same connected component of $G[S]$. Now that $G[B_1]+(x_1,c_1)$ is a biconnected $K_4$-minor-free graph, hence an SP-graph. It follows that
  Reduction rule \ref{rrule:2cut} applies to $B_1$ and reduces it to a single edge: a contradiction to the fact that the instance is simplified. It follows $C_1\subseteq N_0$.

  As a corollary we know that $\vec{\mc{B}}_{F_{\alpha}}$ contains no other leaf block and thus it is a path. It remains to bound the size of $V_{\alpha}$. Since $C_1\subseteq N_0$ and $\{c_1,c\}$ forms a boundary of $C_1$, whenever $|C_1|>4$ Reduction rule \ref{rrule:2cut} applies, contradiction. Hence $|V_{\alpha}|=|B_1|+|C_1\setminus \{c_1\}|$ and combining the bound given by Lemma \ref{lem:cst-size-block}, we obtain the upper bound $\gamma(9) + 3$.

  \noindent (b) $B_1$ is a trivial block (i.e. an edge)\\
  W.l.o.g. $\vec{\mc{B}}_{F_{\alpha}}$ does not contain a nontrivial leaf block. Consider the remaining part of $G_{\alpha}$, i.e. $C_1:=V_{\alpha}\setminus \{x_1\}$. Here we claim that $|N_S(C_1)|\leq 1$. Suppose the contrary. By Lemma \ref{lem:K4-3+1}, we have $|N_S(V_{\alpha})|\leqslant 2$. Hence considering the case when $N_S(x_1)=N_S(C_1)=\{u,v\}$ is sufficient. It remains to see that $u$ and $v$ belong to the same connected component of $G[S]$, and $G[V_{\alpha}\cup S]$ contains a $K_4$-subdivision with $x_1,C_1,u,v$ as branching nodes, a contradiction.

  As a corollary we know that $\vec{\mc{B}}_{F_{\alpha}}$ contains no other leaf block and thus it is a path. It remains to bound the size of $V_{\alpha}$. Consider the case when every block of $\vec{\mc{B}}_{F_{\alpha}}$ trivial, i.e. $G_{\alpha}$ is a path. From the argument of the previous paragraph, we know that $|N_S(C_1)|\leq 1$ and $N_S(C_1)\subseteq N_S(x_1)$. Since the instance is reduced with respect to 1-Boundary rule~\ref{rrule:deg2} and Chandelier rule~\ref{rrule:consec}, we can conclude
  that $|V_{\alpha}|\leqslant 4$.

  Now consider the case $\vec{\mc{B}}_{F_{\alpha}}$ contains a nontrivial block and let $B_2$ be the nontrivial block which is farthest from $c$. Since $\vec{\mc{B}}_{F_{\alpha}}$ is a path, it can be partitioned into two subpaths: the one starting from the cut node $c$ to the block $B_2$ and the remaining part. Let $G_0$ and $G_1$ be the associated subgraphs of $G_{\alpha}$, i.e. containing the vertices which appear in each subpath as part of a block or as a cut node. As every block of $G_1$ is trivial, the bound in the previous paragraph applies and $|G_1|\leq 4$. Observe that the bound obtained in (a) applies to $G_0$: to be precise, applies to the graph obtained from $G_{\alpha}$ by contracting $G_1$ into a single vertex. Hence we get the desired bound $|V_{\alpha}|\leq |G_1|+|G_2| =\gamma(9)+ 7$.
\end{proof}

\begin{reminder}{Lemma \ref{lem:marked-cst-size}}
  Let $(G,S,k)$ be a simplified instance of \KCC{} and let
  $\alpha$ be a marked node of the extended SP-decomposition
  $(T,\mathcal{X})$ of $G[F]$, then $|V_{\alpha}|\leqslant
  c_1=12(\gamma(9)+2c_0)$.
\end{reminder}
\begin{proof}
  We consider each possible type of node separately. Recall
  that since $\alpha$ is marked, the neighbourhood
  $N_S(V_{\alpha})$ belongs to a single biconnected component and
  $G[S\cup V_{\alpha}]$ is $K_4$-minor-free. When \(\alpha\) is a
  \emph{cut node}, Lemma~\ref{lem:simplepath} directly provides the
  bound. We now consider the remaining cases.

\smallskip
\noindent
\textbf{(1) $\alpha$ is an edge node:} By the construction of an
extended SP-decomposition $(T,\mathcal{X})$, any child of $\alpha$
is a cut node. Since $\alpha$ can have at most two children, Lemma
\ref{lem:simplepath} implies $|V_{\alpha}|\leqslant 2c_0$.

\smallskip
\noindent
\textbf{(2) $\alpha$ is a P-node: } Recall that we have
$|V^B_{\alpha}|<\gamma(9) $ by Lemma \ref{lem:cst-size-block} and
$\alpha$ has at most two attachment vertices by Lemma
\ref{lem:2attach}. Each attachment vertex of $\alpha$ either
belongs to $N_1\cup N_2$ or is a cut vertex. Hence we can apply
the bound on cut node size given by Lemma \ref{lem:simplepath}. It
follows that $|V_{\alpha}|\leqslant \gamma(9) + 2c_0$.

\smallskip
\noindent
\textbf{(3) $\alpha$ is an S-node: } Let $\beta_1,\ldots, \beta_q$
be the children of $\alpha$ and denote bye $x_1\dots x_{q+1}$ the
vertices such that for $1\leqslant j\leqslant q$,
$X_{\beta_j}=\{x_j,x_{j+1}\}$. Since every child of an S-node is
either a P-node or an edge node, from case 1 and 2 we have
$|V_{\beta_j}|\leq \gamma(9) + 2c_0$. We now prove that if
$q\geqslant 13$, then either the instance is not simplified or
$G[S\cup V_{\alpha}]$ contains $K_4$ as a minor. Since the lemma
holds trivially if every \(V_{\beta_{j}}\) has at most four
vertices, in the rest of the proof we assume without loss of
generality that for each P-node \(\beta_{j}\) which we consider,
\(|V_{\beta_{j}}|>4\).

\begin{claim}\label{cl:atleast}
  For $1\leqslant j\leqslant q-1$, let $Z_j:= V_{\beta_{j}} \cup
  V_{\beta_{j+1}}$. Then $Z_j\setminus \partial_{F}(Z_j)$ contains
  at least one vertex in $N_1\cup N_2$.
\end{claim}
\begin{proofof}
Suppose one of $\beta_j$ and $\beta_{j+1}$, say $\beta_j$, is a P-node. By Lemma \ref{lem:withattach}, $Y_{\beta_j}= V_{\beta_j}\setminus X_{\beta_j}$ contains a vertex of $N_1\cup N_2$. If both of $\beta_j$ and $\beta_{j+1}$ are edge nodes, then $x_{j+1}\in N_1\cup N_2$, since otherwise its degree in $G$ is two and we can apply Reduction Rule~\ref{rrule:deg2}, a contradiction.
\end{proofof}

Suppose that $q\geqslant 13$. First, suppose there exists $j$,
$3\leqslant j \leqslant q-2$, such that $\beta_j$ is a P-node. By
Lemma \ref{lem:withattach}, we have $Y_{\beta_j}\cap (N_1\cup
N_2)\neq \emptyset$. On the other hand, Claim \ref{cl:atleast}
says that the subsets $Z_{j-2}$ and $Z_{j+1}$ both contain at
least one vertex in $N_1\cup N_2$ each. Since $G[V^B_{\beta_j}]$
is biconnected  and $G[(S\cup Z_{j-2}
\cup Z_{j+1})\setminus X_{\beta_j}]$ is connected, Lemma
\ref{lem:K4-3+1} applies to these two graphs and there is a
$K_4$-subdivision in $G[S\cup V_{\alpha}]$, a contradiction.

Therefore, we can assume that for every $j$, $3\leqslant j \leq q-2$, $\beta_j$ is an edge node. It follows that
$G[X']$, with $X'=\{x_j:3\leqslant j \leqslant q-2\}$, is a chordless path. Claim \ref{cl:atleast} implies that every internal vertex of $X'$ is an attachment vertex, that is, either it belongs to $N_1\cup N_2$ or it is a cut vertex belonging to some $A(\beta_j)$. We consider the two sets $X_1:=\bigcup_{3\leqslant j\leqslant 6}V_{\beta_j}$ and $X_2:=\bigcup_{8\leqslant j\leqslant 11}V_{\beta_j}$.

\begin{claim}\label{cl:2neighb}
$|N_S(X_1)|\geqslant 2$ and $|N_S(X_2)|\geqslant 2$.
\end{claim}
\begin{proofof}
  Consider $X'_1=\{x_4,x_5,x_6,x_7\}$. Suppose that every vertex
  on $X'_1$ belongs $N_1\cup N_2$. As the instance is reduced with
  respect to Rule \ref{rrule:consec} and $|X'_1|=4$, clearly we
  have $|N_S(X_1)|\geqslant 2$. Hence we may assume there exists a
  cut vertex $x\in X'_1$ and let $\alpha_x$ be the cut node of
  $(T,\mathcal{X})$ with $X_{\alpha_x}=\{x\}$. By Lemma
  \ref{lem:simplepath}, there is only one leaf block $B_x$ in
  $G_{\alpha_x}$. If $B_x$ is a single edge, $B_x$ contains a
  pendant vertex $y$. Observe that $N_S(y)=2$ and the claim
  holds. Consider the case $B_x$ is a nontrivial block. By Lemma
  \ref{lem:withattach}, $B_x$ contains a vertex $y\neq c$ in
  $N_1\cup N_2$, where $c$ is the unique cut vertex contained in
  $B_x$. In fact, $B_x$ does not contain $z\neq y$ such that $z\in
  N_1\cup N_2$, since otherwise $|\partial_G(B_x)|\geq 3$ and
  applying Lemma~\ref{lem:K4-3+1} on $Y:=B_x$, $W:=C\cup
  (V_{\alpha}\setminus B_x)$ (with $C$ the connected component of
  $N_S(V_{\alpha})$) witnesses $K_4$-subdivision in $G[S\cup
  V_{\alpha}]$, a contradiction. So we have
  $\partial_G(B_x)=\{c,y\}$. As we assume that the instance is
  reduced, in particular with respect to Reduction rule
  \ref{rrule:2cut}, and $B_x$ is a nontrivial block, we conclude
  that $B_x$ is a $\theta_3$ with $c$ and $y$ as subdividing
  nodes. On the other hand, it is not difficult to see that
  $G[S\cup V_{\alpha}]$ contains a $c,y$-path $P$ avoiding
  $B_x$. It remains to observe that $G[B_x\cup P]$ is a
  $K_4$-model, a contradiction.
\end{proofof}

If $|N_S(V_{\alpha})|\geqslant 3$\, then Lemma \ref{lem:K4-3+1} applies to the biconnected component of $N_S(V_{\alpha})$ and $V_{\alpha}$, thus we obtain a $K_4$-subdivision, a contradiction. If $|N_S(V_{\alpha})|=2$, then $N_S(X_1)=N_S(X_2)$ and $G[S\cup V_{\alpha}]$ contains a $K_4$-model with branching nodes being the following four connected subsets, a contradiction: $X_1$, $X_2$, each of the two vertices of $N_S(X)$. That is, we have a $K_4$-model in $G[S\cup V_{\alpha}]$ whenever $q\geqslant 13$. Therefore, we have $q\leqslant 12$ if $\alpha$ is marked.
\end{proof}

\section{Deferred proof of Lemma \ref{lem:poly-branching}}

\begin{reminder}{Lemma \ref{lem:poly-branching}} Let $(G,S,k)$ be a
  simplified instance of \KCC{} and let $\alpha$ be a lowest
  unmarked node of $(T,\mathcal{X})$
  of $G[F]$.  In polynomial time, one can find
\begin{enumerate}
\vspace{-0.2cm}
\item[(a)] a path $X$ of size at most $2c_1$ satisfying the conditions
  of line~\ref{l:Xbr2} (resp. line~\ref{l:Xbr3}) if the test at
  line~\ref{l:br2-test} (resp.~\ref{l:br3-test}) succeeds;
\vspace{-0.2cm}
\item[(b)] a subset $X\subseteq V_{\alpha}$ of size bounded by
  $2c_1$ satisfying the condition of line~\ref{l:Xbr1} if the test at
  line~\ref{l:br1-test} succeeds;
\end{enumerate}
\end{reminder}
\begin{proof}

  Suppose that $\alpha$ is a cut node. If the test at line
  \ref{l:br2-test} or at line \ref{l:br3-test} succeeds, then
  there are two children $\beta_1$ and $\beta_2$ of $\alpha$ such
  that $X:=V_{\beta_1}\cup V_{\beta_2}$ satisfies the conditions
  of line~\ref{l:br2} or line~\ref{l:br3}, respectively. In case
  of (b), the proof of Lemma \ref{lem:simplepath} shows that if
  $\alpha$ has two children $\beta_1$ and $\beta_2$, then the
  subgraph $G[X\cup S]$ contains $K_4$ as a minor, where
  $X:=V_{\beta_1}\cup V_{\beta_2}$. With the bound provided by
  Lemma \ref{lem:marked-cst-size}, now it suffices to argue that
  $X$ is a connected set. We claim that $c\in X_{\beta_1}\cap
  X_{\beta_2}$. Indeed, $\beta_i$ is either a P-node or an edge
  node. Obviously, $c\in X_{\beta_i}$ if $\beta_i$ is an edge
  node. If $\beta_i$ is a P-node, recall that this is the root
  node of the canonical SP-tree $(T^B, \mc{X}^B)$ from which
  $\beta_i$ is inherited. Since $(c,G^B_{\beta_i})\in
  E(\vec{\mc{B}}_G)$, the construction of $(T^B, \mc{X}^B)$
  requires that $c\in X_{\beta_i}$. As a result, $c\in
  X_{\beta_1}\cap X_{\beta_2}$ and the subgraph $G[V_{\beta_1}\cup
  V_{\beta_2}]$ is connected.

  If $\alpha$ is an edge node, $\alpha$ can have at most two
  children, all of which are cut nodes. Take
  $X=V_{\alpha}$. Since every child of $\alpha$ is marked already,
  the bound of Lemma \ref{lem:simplepath} holds and $|X|\leqslant
  2c_0$. In $G[X]$, one can identify a path or a subset satisfying
  the condition (a) or (b).

If $\alpha$ is a P-node, let $\beta_1$ and $\beta_2$ be its two children. By Lemma \ref{lem:marked-cst-size}, we know that $|V_{\beta_1}|, |V_{\beta_2}|\leqslant c_1$. Take $X=V_{\alpha}$. In $G[X]$, one can identify a path or a subset satisfying the condition (a) or (b) if this is the case.

Let us consider the case when $\alpha$ is an S-node with
$\beta_1,\ldots , \beta_q$ as its children. Suppose that there are
$u,v \in V_{\alpha}\cap (N_1\cup N_2)$ which have neighbors in
distinct connected components of $G[S]$. 
Then there exist $1\leqslant k < k' \leqslant q$ such that $u\in
V_{\beta_k}$ and $v\in V_{\beta_{k'}}$. Choose $k$ and $k'$ such
that $k'-k$ is minimized. We claim that $k'-k\leqslant 2$. Suppose
not. Then we can find an alternative vertex $w \in Z_{k+1}\cap
(N_1\cup N_2)$ due to Claim \ref{cl:atleast} in the proof of Lemma
\ref{lem:marked-cst-size} and decrease $k'-k$, a
contradiction. Therefore, there exists $k$ such that
$X:=V_{\beta_k}\cup V_{\beta_{k+1}} \cup V_{\beta_{k+2}}$ contains
$u,v$. It remains to observe that $|X|\leqslant 3\times (\gamma(9)
+2c_0)$ and we can find a path $P$ between $u$ and $v$ within $X$,
satisfying (a). The proof remains the same when there are $u,v \in
V_{\alpha}\cap (N_1\cup N_2)$ with $bc_S(u)\neq bc_S(v)$. On the
other hand if the test at line~\ref{l:br1-test} succeeds, the
proof of Case (3) in Lemma \ref{lem:marked-cst-size} shows one can
find a bounded-size subset $X$. Indeed, if $q\leqslant 12$, one
can take $X:=V_{\alpha}$ and observe that $|X|\leqslant
12(\gamma(9) + 2c_0)\leqslant 2c_1$. If $q\geqslant 13$, take $X:=
\bigcup_{j=1}^{13} V_{\beta_j}$ and observe that $|X|\leqslant
13(\gamma(9) +2c_0)\leqslant 2c_1$.
\end{proof}

\end{document}